\documentclass[letter,11pt]{article}

\usepackage{bm}
\usepackage{algorithm}
\usepackage{amsthm}
\usepackage[noend]{algpseudocode}
\usepackage{array}
\usepackage{enumitem}
\usepackage{geometry}
\usepackage{graphicx}
\usepackage{listings}
\usepackage{mathtools}
\usepackage{subcaption}
\usepackage{tabularx}
\usepackage{tcolorbox}
\usepackage{varwidth}

\newcommand{\myparagraph}[1]{\smallskip\noindent {\bf #1.}}

\newtheorem{theorem}{Theorem}[section]
\newtheorem{lemma}{Lemma}[section]
\newtheorem{definition}{Definition}[section]

\newtheorem{remark}{Remark}[section]

\algrenewcommand\algorithmicindent{1em}

\algnewcommand{\algcomment}[1]{\hfill{\color{purple!40!black}\emph{// #1}}}    
\algnewcommand{\alglinecomment}[1]{{\color{purple!40!black}\emph{// #1}}}      

\algnewcommand\alglocal{\textbf{local }}                            
\algnewcommand\algreturn{\textbf{return }}                          
\algnewcommand\algeach{\textbf{each }}                              
\algnewcommand\algretire{\textbf{retire process}}                   

\algnewcommand\algmodassign{\textbf{write}}                         
\algnewcommand\algwritemod[2]{\algmodassign(#1, #2)}          		
\algnewcommand\algto{\textbf{to }}									
\algnewcommand\algis{\textbf{is}}                                   
\algnewcommand\algnot{\textbf{not}}                                 
\algnewcommand\algin{\textbf{in}}                                   
\algnewcommand\algempty{\textbf{empty}}                             

\algnewcommand\algand{\textbf{ and }}                               
\algnewcommand\algor{\textbf{ or }}                                 
\algnewcommand\algassign{\ensuremath{\gets} }                       

\algnewcommand\algtrue{\textbf{true}}                               
\algnewcommand\algfalse{\textbf{false}}                             
\algnewcommand\algnull{\ensuremath{\perp}}                          

\algnewcommand\algarray[1]{\textnormal{array}\ensuremath{\langle}#1\ensuremath{\rangle}}    
\algnewcommand\algmod[1]{\textbf{mod}\ensuremath{\langle}#1\ensuremath{\rangle}}            

\algnewcommand\algorithmicwith{\textbf{with}}
\algnewcommand\algorithmicread{\textbf{read}}
\algnewcommand\algorithmicas{\textbf{as}}
\algnewcommand\Read[2]{\State \alglocal #2 \algassign \algorithmicread(#1)}
\algnewcommand\EndRead{}

\algnewcommand\algorithmicinparallel{\textbf{in parallel}}
\algblockdefx[PARFOR]{ParallelFor}{EndParallelFor}[1]
{\algorithmicfor\ #1\ \algorithmicdo\ \algorithmicinparallel}
{\algorithmicend\ \algorithmicfor}

\makeatletter
\ifthenelse{\equal{\ALG@noend}{t}}%
{\algtext*{EndParallelFor}}
{}%
\makeatother

\mathtoolsset{showonlyrefs}

\newcommand{\expct}[1]{\mathop{{}\mathrm{\bf E}}\left[ #1 \right]}
\newcommand{\probb}[1]{\mathop{{}\mathrm{\bf P}}\left[ #1 \right]}
\newcommand{\setsize}[1]{\left|#1\right|}
\newcommand{\bigO}[1]{O\left( #1 \right)}

\providecommand{\set}[1]{\ensuremath{\left\{#1\right\}}}

\newenvironment{proofsketch}{%
  \proof}{\endproof}

\newcommand{\mpram}{$\mathsf{TRAM}$}

\newcommand{\PRAM}[0]{\ensuremath{\mathsf{PRAM}}}

\newcommand{\crcwpram}{CRCW \ensuremath{\mathsf{PRAM}}}

\newcommand{\multipram}[0]{\emph{multiprefix} CRCW \ensuremath{\mathsf{PRAM}}\renewcommand{\multipram}[0]{multiprefix CRCW \ensuremath{\mathsf{PRAM}}}}

\newcommand{\Roundbased}{Round-synchronous}
\newcommand{\roundbased}{round-synchronous}
\newcommand{\roundBased}{round synchronous}

\newcolumntype{P}[1]{>{\centering\arraybackslash}p{#1}}

\begin{document}

  \title{Parallel Batch-dynamic Trees via Change Propagation}

  \author{\normalsize Umut A. Acar \\ \normalsize Carnegie Mellon University \\ \normalsize umut@cs.cmu.edu \and \normalsize Daniel Anderson \\ \normalsize Carnegie Mellon University \\ \normalsize dlanders@cs.cmu.edu \and \normalsize Guy E. Blelloch \\ \normalsize Carnegie Mellon University \\ \normalsize guyb@cs.cmu.edu \and \normalsize Laxman Dhulipala \\ \normalsize Carnegie Mellon University \\ \normalsize ldhulipa@cs.cmu.edu \and \normalsize Sam Westrick \\ \normalsize Carnegie Mellon University \\ \normalsize swestric@cs.cmu.edu}
  \date{}

  \maketitle
  \begin{abstract}
The dynamic trees problem is to maintain a forest subject to edge insertions and deletions while facilitating queries such as connectivity, path weights, and subtree weights. Dynamic trees are a fundamental building block of a large number of graph algorithms.
Although traditionally studied in the single-update setting, dynamic algorithms capable of supporting batches of updates are increasingly
relevant today due to the emergence of rapidly evolving dynamic
datasets.
Since processing updates on a single processor is often unrealistic
for large batches of updates, designing parallel batch-dynamic algorithms that
achieve provably low span is important for many applications.

In this work, we design the first work-efficient parallel batch-dynamic algorithm for dynamic trees that is capable of supporting both path queries and subtree queries, as well as a variety of non-local queries. Previous work-efficient dynamic trees of Tseng et al.\ were only capable of handling subtree queries under an invertible associative operation [\textit{ALENEX'19}, (2019), pp.~92--106]. To achieve this, we propose a framework for
algorithmically dynamizing static \roundbased{} algorithms that allows us to
obtain parallel batch-dynamic algorithms with good bounds on their
work and span. In our framework, the algorithm designer can apply the technique to any suitably defined static algorithm.
We then obtain theoretical guarantees for algorithms in our framework
by defining the notion of a computation distance between
two executions of the underlying algorithm.

Our dynamic trees algorithm is obtained by applying our dynamization framework to the parallel tree contraction algorithm of Miller and Reif [\textit{FOCS'85}, (1985), pp.~478--489], and then performing a novel analysis of the computation distance of this algorithm under batch updates. We show that $k$ updates can be performed in $O(k\log(1+n/k))$ work in expectation, which matches the algorithm of Tseng et al.~while providing support for a substantially larger number of queries and applications.
\end{abstract}
  \thispagestyle{empty}
  
  \clearpage
  \pagenumbering{arabic} 

  \clearpage

  \section{Introduction}

The dynamic trees problem, first posed by Sleator and Tarjan~\cite{sleatorta83} is to maintain a forest of trees subject to the insertion and deletion of edges, also known as \emph{links} and \emph{cuts}. Dynamic trees are used as a building block in a multitude of applications, including maximum flows~\cite{sleatorta83}, dynamic connectivity and minimum spanning trees~\cite{frederickson1985data}, and minimum cuts~\cite{karger2000minimum}, making them a fruitful line of work with a rich history. There are a number of established sequential dynamic tree algorithms, including link-cut trees~\cite{sleatorta83}, top trees~\cite{tarjan2005self}, Euler-tour trees~\cite{henzingerki99}, and rake-compress trees~\cite{acar2005experimental}, all of which achieve $O(\log(n))$ time per operation.

Since they already perform such little work, there is often little to gain by processing single updates in parallel, hence parallel applications typically process \emph{batches} of updates. We are therefore concerned with the design of parallel \emph{batch-dynamic} algorithms.

The batch-dynamic setting extends classic dynamic algorithms to accept
batches of updates.  By applying batches it is often possible to
obtain significant parallelism while preserving work efficiency.
However, designing and implementing dynamic algorithms for problems is
difficult even in the sequential setting, and arguably even more so in
the parallel setting.

The goals of this paper are twofold. First and foremost, we are interested in designing a parallel batch-dynamic algorithm for dynamic trees that supports a wide range of applications. On another level, based on the observation that parallel dynamic algorithms are usually quite complex and difficult to design, analyze, and implement, we are also interested in easing the design process of parallel batch-dynamic algorithms as a whole.
To this end, we propose a framework for algorithmically dynamizing static parallel algorithms to obtain
efficient parallel batch-dynamic algorithms. The framework takes any
algorithm implemented in the \roundBased{} parallel model, and
automatically dynamizes it. We then define a cost model that captures the \emph{computation distance} between two executions of the static algorithm which allows us to bound the runtime of dynamic updates. There are several benefits of using algorithmic dynamization, some more
theoretical some practical:
\begin{enumerate}[leftmargin=12pt]
  \item
  Proving correctness of a batch dynamized algorithm relies simply on the correctness of the parallel
  algorithm, which presumably has already been proven.
  \item It is easy to implement different classes of updates.  For
  example, for dynamic trees, in addition to links and cuts, it is
  very easy to update edge weights or vertex weights for supporting
  queries such as path length, subtree sums, or weighted diameter.
  One need only change the values of the weights and propagate.
  \item Due to the simplicity of our approach, we believe it is likely
  to make it easier to program parallel batch-dynamic algorithms, and
  also result in practical implementations.
\end{enumerate}

\noindent Using our algorithmic dynamization framework, we obtain a parallel batch-dynamic algorithm for rake-compress trees that generalizes the sequential data structure work efficiently without loss of generality. Specifically, our main contribution is the following theorem.

\begin{theorem}  \label{theorem:rctree}
  The following operations can be supported on a bounded-degree dynamic tree of size $n$ using
  the \crcwpram: Batch insertions and deletions of $k$ edges in
  $O(k\log(1+n/k))$ work in expectation and $O(\log(n) \log^*(n))$
  span  w.h.p. Independent parallel connectivity, subtree-sum, path-sum, diameter, lowest common ancestor, center, and median queries in $O(\log n)$ time per query w.h.p.
\end{theorem}

\noindent This algorithm is obtained by dynamizing the parallel tree contraction algorithm of Miller and Reif~\cite{miller1985parallel} and performing a novel analysis of the computation distance of the algorithm with respect to any given input.

As some evidence of the general applicability of algorithmic dynamization, in addition to dynamic trees we consider some other applications of the technique.  We consider map-reduce based computations in Appendix~\ref{sec:map-reduce}, and dynamic sequences with
cutting and joining in Appendix~\ref{sec:list-contraction}.    This leads to another solution to the batch-dynamic sequences
problem considered in~\cite{tseng2018batch}.  We
believe the solution here is much simpler, while maintaining the same
work bounds.
To summarize, the main contributions of this paper are:
\begin{enumerate}[leftmargin=12pt]
  \item An algorithmic framework for dynamizing \roundbased{} parallel algorithms, and a cost model for analyzing the performance of algorithms resulting from the framework
  \item An analysis of the computation distance of Miller and Reif's tree contraction algorithm under batch edge insertions and deletions in this framework, which shows that it can be efficiently dynamized
  \item The first work-efficient parallel algorithm for batch-dynamic trees that supports subtree queries, path queries, and non-local queries such as centers and medians.
\end{enumerate}

\paragraph{Technical overview.}
A \emph{\roundbased} algorithm consists
of a sequence in rounds, where a round executes in parallel across a
set of processes, and each process runs a sequencial \emph{round
  computation} reading and writing from shared memory and doing local
computation. The \roundBased{} model is similar to Valiant's
well-known Bulk Synchronous Parallel (BSP) model~\cite{valiant90},
except that communication is done via shared memory.

 The algorithmic dynamization works by running the \roundbased{}
algorithm while tracking all write-read dependences---i.e., a
dependence from a write in one round to a read in a later round.
Then, whenever a batch of changes are made to the input, a
\emph{change propagation} algorithm propagates the changes through the
original computation only rerunning round computations if the values
they read have changed.  This can be repeated to handle multiple batch
changes. We note that depending on the algorithm, changes to the input
could drastically change the underlying computation, introducing new
dependencies, or invalidating existing ones. Part of the novelty of
this paper is bounding both the work and span of this update process.

The idea of change
propagation has been applied in the sequential setting and used to
generate efficient dynamic
algorithms~\cite{acar2002adaptive,acar2004dynamizing}.  The general
idea of parallel change propagation has also been used in various practical
systems~\cite{Condie+10,Gunda+10,BhatotiaWiRoAcPa11,Murray+2013,PengDa10} but none of them have been analyzed
theoretically.
%
To capture the cost of running the change propagation algorithm for a
particular parallel algorithm and class of updates we define a
\emph{computational distance} between two computations, which
corresponds to the total work of the round computations that differ in
the two computations.    The \emph{input configuration} for a
computation consists of the input $I$, stored in shared memory, and an
initial set of processes $P$.
 We show the following bounds, where the
\emph{work} is the sum of the time of all round computations, and
\emph{span} is the sum over rounds of the maximum time of any
round computation in that round.

\begin{theorem}\label{thm:introchangeprop}
  Given a \roundbased{} algorithm $A$ that with input configuration $(I,P)$ 
  does $W$ work in $R$ rounds and $S$ span, then
\begin{enumerate}[leftmargin=12pt]
\item the initial run of the algorithm with tracking takes $O(W)$ work in expectation and $O(S +
  R \log W)$ time w.h.p.,
\item running change propagation from input configuration $(I,P)$ to 
  configuration $(I',P')$ takes $O(W_{\Delta} + R')$ work in expectation and
  $O(S'+ R' \log W')$ time w.h.p., where $W_{\Delta}$ is
  the computation distance between the two configurations, and $S'$,$R'$,$W'$ are the maximum span,
  rounds and work for the two configurations,
\end{enumerate}
all on the \crcwpram{} model.
\end{theorem}
\noindent
We show that the work can be reduced to $O(W_\Delta)$, and that the $\log W$ and $\log W'$ terms can be reduced to $\log^* W$ when the
\roundbased{} algorithms have certain restrictions that are satisfied by all of our example algorithms, including our main result on dynamic trees.
We also present similar results in other parallel models of computation.

Once we have our dynamization framework and cost model, we use it to develop an algorithm for dynamic trees that support a broad set of
queries including subtree sums, path queries, lowest common ancestors, diameter, center, and
median queries.  This significantly improves over previous work on
batch-dynamic Euler tour trees~\cite{tseng2018batch}, which only support subtree sums, and only when the ``summing'' function has an inverse
(e.g. Euler-tour trees cannot be used to take the maximum over a
subtree).

Our batch-dynamic trees algorithm is based on the simple and
elegant tree contraction algorithm of Miller and Reif
(MR)~\cite{miller1989parallel}.
Previous work showed that in the sequential setting,
this process can be used to generate a \emph{rake-compress (RC) tree} (or
forest), which supports the wide collection of queries mentioned above, all in logarithmic time, w.h.p.~\cite{acar2005experimental}.
Our approach generalizes this sequential algorithm to allow for
batches of edge insertions or deletions, work efficiently in parallel.
The challenge is
in analyzing the computational distance implied by batch updates
in the parallel batch-dynamic setting. In Section~\ref{sec:tree-contraction} we do just that, and obtain the following result:

\begin{theorem}\label{thm:introMR}
In the \roundBased{} model the MR algorithm does $O(n)$ work in
expectation and has $O(\log n)$ rounds and span w.h.p.
Furthermore, given forests $T$ with $n$ vertices, and $T'$ with $k$
modifications to the edge list of $T$, the computational distance of the
\roundbased{} MR algorithm on the two inputs is $O(k\log(1 + n/k))$ in expectation.
\end{theorem}
\noindent
The first sentence follows directly from the original MR analysis.
The second is one of the contributions of this paper.   The bounds can
then be plugged into Theorem~\ref{thm:introchangeprop} to show
that a set of  $k$ edges can be inserted or deleted in a batch in
$O(k\log(1 + n/k))$ work in expectation and $O(\log^2 n)$ span w.h.p.
The span can be improved to $O(\log n \log^* n)$ w.h.p. on the \crcwpram{} model. The last step in obtaining our dynamic trees framework is to plug the dynamized tree contraction algorithm into the RC trees framework~\cite{acar2005experimental} (see Section~\ref{sec:rc-trees}).

  \paragraph{Related Work.}
%
In order to better exploit parallelism, work-efficient \emph{parallel batch-dynamic} algorithms, i.e.\ algorithms that perform a batch of updates work efficiently and in low span, have gained recent attention and have been developed for several specific problems including computation of Voronoi diagrams~\cite{acarcohutu11pd}, incremental
connectivity~\cite{simsiri2016work}, Euler-Tour
trees~\cite{tseng2018batch}, and for fully dynamic
connectivity~\cite{acar2019parallel}. Parallel batch-dynamic algorithms have also been recently studied in the MPC model~\cite{italiano2019dynamic,dhulipala2020parallel}.
These works show that batch-dynamic algorithms can achieve tight work-efficiency bounds without sacrificing parallelism.

A sequential version of change propagation was initially developed in
2004~\cite{acar2004dynamizing} and has lead to the development of a sequential dynamic-tree data structure capable of supporting both
subtree and path queries~\cite{acar2005experimental}.
Our results on parallel batch-dynamic trees generalize the sequential results~\cite{acar2004dynamizing,acar2005experimental} to handle batches of changes work efficiently and in parallel without leading to any loss of generality.

Standalone algorithms for dynamic tree contraction have previously been proposed, but are inefficient and not fully general. In particular, Reif and Tate~\cite{reif1994dynamic} give an algorithm for parallel dynamic tree contraction that can process a batch of $k$ leaf insertions or deletions in $O(k\log(n))$ work. Unlike our algorithm, theirs is not work efficient, as it performs $\Omega(n\log(n))$ for batches of size $\Omega(n)$, and it can only modify the tree at the leaves.

  \section{Preliminaries}

\subsection{Parallel Models}\label{subsec:model}

\myparagraph{Parallel Random Access Machine ($\bm{\PRAM{}}$)}
The parallel random access machine (\PRAM{}) model is a classic
parallel model with $p$ processors that work in lock-step, connected
by a parallel shared-memory~\cite{jaja1992introduction}. In this paper
we primarily consider the Concurrent-Read Concurrent-Write model
(\crcwpram{}), where memory locations are allowed to be concurrently
read and concurrently written to. If multiple writers write to the
same location concurrently, we assume that an arbitrary writer wins.
We analyze algorithms on the \crcwpram{} in terms of their \emph{work}
and \emph{span}. The span of an algorithm is the minimum running time
achievable when arbitrarily many processors are available. The work is
the product of the span and the number of processors.


\myparagraph{Binary Forking Threaded Random Access Machine ($\bm{\mpram{}}$)}
The threaded random access machine (\mpram{}) is closely related to the \PRAM{}, but more closely models current machines and programming paradigms. In the binary forking \mpram{} (binary forking model for short), a process
can \emph{fork} another process to run in parallel, and can \emph{join} to wait for all forked calls to complete. In the binary forking model, the \emph{work} of an algorithm is the total number of instructions it performs,
and the \emph{span} is the longest chain of sequentially dependent
instructions. This model can work-efficiently cross-simulate a
\crcwpram{} equipped with the same atomic instructions, and hence all
work bounds stated are valid in both. Additionally, an algorithm with
$W$ work and $S$ span on the \mpram{} can be executed on a $p$
processor \PRAM{} in time $O(W/p + S)$ \cite{brent74}.

\subsection{Parallel Primitives}\label{subsec:parallelprimitives}
The following parallel procedures are used throughout the paper.
\emph{Scan} takes as input an array $A$ of length $n$, an associative
binary operator $\oplus$, and an identity element $\bot$ such that
$\bot \oplus x = x$ for any $x$, and returns the array
$(\bot, \bot \oplus A[0], \bot \oplus A[0] \oplus A[1], \ldots, \bot \oplus_{i=0}^{n-2} A[i])$
as well as the overall sum, $\bot \oplus_{i=0}^{n-1} A[i]$.
Scan can be done in $O(n)$ work and $O(\log n)$ span (assuming $\oplus$
takes $O(1)$ work)~\cite{jaja1992introduction} on the \crcwpram{}, $O(\log(n))$ span in the binary forking model.

\emph{Filter} takes an array $A$ and a predicate $f$ and returns a new
array containing $a \in A$ for which $f(a)$ is true, in the same order
as in $A$.  Filter can both be done in $O(n)$ work and $O(\log n)$
span on the \crcwpram{} (assuming $f$ takes $O(1)$ work)~\cite{jaja1992introduction}, $O(\log(n))$ span in the binary forking model.
The \emph{Approximate Compaction} problem is similar to a Filter. It
takes an array $A$ and a predicate $f$ and returns a new array
containing $a \in A$ for which $f(a)$ is true where some of the
entries in the returned array can have a null value. The total size of
the returned array is at most a constant factor larger than the number
of non-null elements.  Gil et al.~\cite{gil1991towards} describe a parallel
approximate compaction algorithm that uses linear space and achieves
$O(n)$ work and $O(\log^* (n))$ span w.h.p.\ on the \crcwpram{}.

A \emph{semisort} takes an input array of
elements, where each element has an associated key and reorders the
elements so that elements with equal keys are contiguous, but elements
with different keys are not necessarily ordered.  The purpose is to
collect equal keys together, rather than sort them. Semisorting a
sequence of length $n$ can be performed in $O(n)$ expected work and
$O(\log n)$ depth w.h.p.\ on the \crcwpram{} and in the binary forking
model assuming access to a uniformly random hash function mapping
keys to integers in the range $[1,n^{O(1)}]$~\cite{gu2015top}.

  \section{Dynamization Framework}

\subsection{\Roundbased{} algorithms}

In this framework, we consider dynamizing algorithms that are \emph{\roundBased{}}. The \roundBased{} framework encompasses a range of classic BSP~\cite{valiant90} and \PRAM{} algorithms. A \roundbased{} algorithm consists of $M$ \emph{processes}, with process IDs bounded by $O(M)$. The algorithm performs sequential rounds in which each active process executes, in parallel, a \emph{round computation}.  At the end of a round,
any processes can decide to \emph{retire}, in which case they will no longer execute in any future round. The algorithm terminates once there are no remaining active processes---i.e., they have all retired. Given a fixed input, \roundbased{} algorithms must perform deterministically. Note that this does not preclude us from implementing randomized algorithms (indeed, our dynamic trees algorithm is randomized), it just requires that we provide the source of randomness as an input
to the algorithm, so that its behavior is identical if re-executed.
An algorithm in the \roundBased{} framework is defined in terms of a procedure \textproc{ComputeRound}$(r,p)$, which performs the computation of process $p$ in round $r$. The initial run of a \roundbased{} algorithm must specify the set $P$ of initial process IDs.

\myparagraph{Memory model} Processes in a \roundbased{} algorithm may read and write to local memory that is not persisted across rounds. They also have access to a \emph{shared memory}. The input to a \roundbased{} is the initial contents of the shared memory. Round computations can read and write to shared memory with the condition that writes do not become visible until the end of the round. We require that reads only access shared locations that have been written to, and that locations are only written to once, hence concurrent writes are not permitted. The contents of the shared memory at termination of an algorithm is considered to be the algorithm's output. Change propagation is driven by tracking all reads and writes to shared memory.

\myparagraph{Pseudocode} We describe \roundbased{} algorithms using the following primitives:
\begin{enumerate}[leftmargin=12pt]
  \item The \algorithmicread{} instruction reads the given shared memory locations and returns their values,
  
  \item The \algmodassign\ instruction writes the given value to the given shared memory location.
  
  \item Processes may retire by invoking the \algretire{} instruction.
  
\end{enumerate}

\myparagraph{Measures} The following measures will help us to analyse the efficiency of \roundbased{} algorithms. For convenience, we define the \emph{input configuration}
of a \roundbased{} algorithm as the pair $(I, P)$, where $I$ is the input to the algorithm (i.e.\ the initial state of shared memory) and $P$ is the set of initial process IDs.

\begin{definition}[Initial work, Round complexity, and Span]
  The \emph{initial work} of a \roundbased{} algorithm on some input configuration $(I, P)$ is the sum of the work performed by all of the computations of each processes over all rounds when given that input. Its \emph{round complexity} is the number of rounds that it performs, and its \emph{span} is the sum of the maximum costs per round of the computations performed by each process.
\end{definition}

\subsection{Change propagation}

Given a \roundbased{} algorithm, a \emph{dynamic update} consists of a change to the input configuration, i.e.\ changing some of the input values in shared memory, and optionally, adding or deleting processes. The initial run and change propagation algorithms maintain the following data:
\begin{enumerate}[leftmargin=12pt]
  \item $R_{r,p}$, the memory locations read by process $p$ in round $r$
  \item $W_{r,p}$, the memory locations written by process $p$ in round $r$
  \item $S_{m}$, the set of round, process pairs that read memory location $m$
  \item $X_{r,p}$, which is \algtrue\ if process $p$ retired in round $r$
\end{enumerate}
Algorithm~\ref{alg:round-based-execute} depicts the procedure
for executing the initial run of a \roundbased{} algorithm before making any dynamic updates.

\begin{algorithm}[h]
  \caption{Initial run}
  \label{alg:round-based-execute}
  \scriptsize
  \begin{algorithmic}[1]
    \Procedure{Run}{$P$}
      \State \alglocal $r \algassign 0$
      \While{$P \neq \emptyset$}
        \ParallelFor{\algeach process $p \in P$}
          \State \textsc{ComputeRound}($r, p$)
          \State\label{code:run-R_r,p} $R_{r,p}$ \algassign $\{$memory locations read by $p$ in round $r\}$
          \State\label{code:run-W_r,p} $W_{r,p}$ \algassign $\{$memory locations written to by $p$ in round $r\}$
          \State\label{code:run-X_r,p} $X_{r,p} \algassign$ (\algtrue\ \algorithmicif\ $p$ retired in round $r$ else \algfalse{})
        \EndParallelFor

        \ParallelFor{\algeach{} $m \in \cup_{p \in P} R_{r,p}$}
          \State\label{code:run-subscribe} $S_m$ \algassign{} $S_m \cup \{ (r, p)\ |\ m \in R_{r,p} \land\ p \in P \}$
        \EndParallelFor

        \State\label{code:run-filter-P} $P \algassign P \setminus \{ p \in P : X_{r,p} = \algtrue \}$
        \State $r \algassign r + 1$
      \EndWhile
    \EndProcedure
  \end{algorithmic}
\end{algorithm}

\noindent To help formalize change propagation, we define the notion of an \emph{affected computation}. The task of change propagation is to identify the affected computations and rerun them. 

\begin{definition}[Affected computation]
  Given a \roundbased{} algorithm $A$ and two input configurations $(I, P)$ and $(I', P')$, the \emph{affected computations} are the round and process pairs $(r,p)$ such that either:
  \begin{enumerate}[leftmargin=12pt]
    \item process $p$ runs in round $r$ on one input configuration but not the other
    \item process $p$ runs in round $r$ on both input configurations, but reads a variable from shared memory that has a different value in one configuration than the other
  \end{enumerate}
\end{definition}

\noindent The change propagation algorithm is depicted in Algorithm~\ref{alg:round-based-change-prop}.

\begin{algorithm}[h]
  \caption{Change propagation}
  \label{alg:round-based-change-prop}
  \scriptsize
  \begin{algorithmic}[1]
    \State \alglinecomment{$U$ = sequence of memory locations that have been modified}
    \State \alglinecomment{$P^+$ = sequence of new process IDs to create}
    \State \alglinecomment{$P^{-}$ = sequence of process IDs to remove}
    \Procedure{Propagate}{$U$, $P^{+}$, $P^{-}$}
    \State \alglocal $D \algassign P^{-}$ \algcomment{Processes that died earlier than before}
    \State \alglocal $L \algassign P^{+}$ \algcomment{Processes that lived longer than before}
    \State \alglocal{} $A$ \algassign{} $\emptyset$ \algcomment{Affected computations at each round}
    \State \alglocal $r \algassign 0$
    \While{$U \neq \emptyset \lor D \neq \emptyset \lor L \neq \emptyset \lor \exists r' \geq r:\ (A_{r
    '} \neq \emptyset)$}
      \State \alglinecomment{Determine the computations that become affected}
      \State \alglinecomment{due to the newly updated memory locations $U$}
      \State\label{code:change-A} $\alglocal{} A' \algassign{} \cup_{m \in U} S_m$
      \ParallelFor{\algeach $r' \in \cup_{(r',p) \in A'} \{ r' \} $}
        \State\label{code:change-bucket-affected-computations} $A_{r'} \algassign{} A_{r'} \cup \{ p\ |\ (r', p) \in A' \}$
      \EndParallelFor
      \State\label{code:change-P} \alglocal{} $P$ \algassign{} $A_r \setminus D$ \algcomment{Processes to rerun}

      \State \alglinecomment{Forget the prior reads of all processes that are}
      \State \alglinecomment{now dead or will be rerun on this round}
      \ParallelFor{\algeach{} $m \in \cup_{p \in P \cup D} R_{r,p}$}\label{code:change-delete-Sm-loop}
        \State\label{code:change-delete-Sm} $S_m$ \algassign{} $S_m \setminus \{ (r, p)\ |\ m \in R_{r,p} \land\ p \in P \cup D \}$
      \EndParallelFor

      \State\label{code:change-remember-X} \alglocal $X^\text{prev} = \{ p \mapsto X_{r,p}\ |\ p \in P \}$
      
      \State \alglinecomment{(Re)run all changed or newly live processes}
      \ParallelFor{\algeach process $p$ \algin\ $P \cup L$}\label{code:change-rerun-loop}
        \State\label{code:change-reexecute} \textsc{ComputeRound}($r, p$)
        \State $R_{r,p}$ \algassign $\{$memory locations read by $p$ in round $r\}$
        \State $W_{r,p}$ \algassign $\{$memory locations written to by $p$ in round $r\}$
        \State\label{code:change-set-X} $X_{r,p} \algassign$ (\algtrue{} \algorithmicif\ $p$ retired in round $r$ else \algfalse{})
      \EndParallelFor

      \State \alglinecomment{Remember the reads performed by processes on this round}
      \ParallelFor{\algeach{} $m \in \cup_{p \in P \cup L} R_{r,p}$}\label{code:change-insert-Sm-loop}
        \State\label{code:change-insert-Sm} $S_m$ \algassign{} $S_m \cup \{ (r, p)\ |\ m \in R_{r,p} \land\ p \in P \cup L \}$
      \EndParallelFor

      \State \alglinecomment{Update the sets of changed memory locations,}
      \State \alglinecomment{newly live processes, and newly dead processes}
      \State\label{code:change-U} $U \algassign \cup_{p \in (P \cup L)} W_{r,p} $
      \State\label{code:change-L'} $L' \algassign \{p \in P\ |\ X^\text{prev}_{p} = \algtrue
      \land X_{r,p} = \algfalse\}$
      \State $L \algassign L\ \cup L' \setminus \{p \in L\ |\ X_{r,p} = \algtrue \}$
      \State $D' \algassign \{p \in P\ |\ X^\text{prev}_{p} = \algfalse \land X_{r,p} = \algtrue\}$
      \State\label{code:change-D} $D \algassign D \cup D' \setminus \{p \in D\ |\ X_{r,p} = \algtrue\}$
      \State $r \algassign r + 1$
    \EndWhile
    \EndProcedure
  \end{algorithmic}
\end{algorithm}

Algorithm~\ref{alg:round-based-change-prop} works by maintaining the affected computations as three disjoint sets, $P$, the set of processes that read a memory location that was rewritten, $L$, processes that outlived their previous self, i.e.\ that retired the last time they ran, but did not retire when re-executed, and $D$, processes that retired earlier than their previous self. First, at each round, the algorithm determines the set of computations that should become affected because of shared memory locations that
were rewritten in the previous round (Lines~\ref{code:change-A}--\ref{code:change-bucket-affected-computations}). These are used to determine $P$, the set of affected computations to rerun this round (Line~\ref{code:change-P}). To ensure correctness, the algorithm must then reset the reads that were performed by the computations that are
no longer alive, or that will be reran, since the set of locations that they read may differ from
last time (Lines~\ref{code:change-delete-Sm-loop}--\ref{code:change-delete-Sm}). Lines~\ref{code:change-rerun-loop}--\ref{code:change-set-X} perform the re-execution of all
processes that read a changed memory location, or that lived longer (did not retire) than
in the previous configuration. The algorithm then subscribes the reads of these computations to
the memory locations that they read (Lines~\ref{code:change-insert-Sm-loop}--\ref{code:change-insert-Sm}).
Finally, on Lines~\ref{code:change-U}--\ref{code:change-D}, the algorithm updates the set of changed memory locations ($U$), the set of computations
that lived longer than their previous self ($L$) and the set of computations that retired earlier
then their previous self ($D$).

  \subsection{Correctness}

In this section, we sketch a proof of correctness of the change propagation algorithm (Algorithm~\ref{alg:round-based-change-prop}).
Intuitively, correctness is assured because
of the write-once condition on global shared memory, which ensures that computations
can not have their output overwritten, and hence do not need to be re-executed unless
data that they depend on is modified.

\begin{lemma}\label{lem:reexecuting-affected-computations}
  Given a dynamic update, re-executing only the affected computations for each round will result in the same output as re-executing all computations on the new input.
\end{lemma}

\begin{proof}
  Since by definition they read the same values, computations that are not affected, if re-executed, would produce the same output as they did the first time. Since all shared memory locations can only be written to once, values written by processes that are not re-executed can not have been overwritten, and hence it is safe to not re-execute them, as their output is preserved. Therefore re-executing only the affected computations will produce the same output as re-executing all computations.
\end{proof}

\begin{theorem}[Consistency]\label{thm:consistency}
  Given a dynamic update, change propagation correctly updates the output of the algorithm.
\end{theorem}

\begin{proofsketch}
  Follows from Lemma~\ref{lem:reexecuting-affected-computations} and the fact that all reads and writes to global shared memory are tracked in Algorithm~\ref{alg:round-based-change-prop}, and since global shared memory is the only method by which processes communicate, all affected computations are identified.
\end{proofsketch}

\subsection{Cost analysis}
\label{sec:analysis}

To analyze the work of change propagation, we need to formalize a notion of \emph{computation distance}. Intuitively, the computation distance between two computations is the work performed by one and not the other. We then show that change propagation can efficiently re-execute the affected computations in work proportional to the computation distance.

\begin{definition}[Computation distance]
  Given a \roundbased{} algorithm $A$ and two input configurations, the \emph{computation distance} $W_\Delta$ between them is the sum of the work performed by all of the affected computations with respect to both input configurations.
\end{definition}

\newcommand{\costcompact}{\ensuremath{\mathcal{C}}}
\newcommand{\overhead}{\ensuremath{\mathcal{C}}}

\begin{theorem}\label{thm:generic_change_prop_bounds}
  Given a \roundbased{} algorithm $A$ with input configuration $(I, P)$ that does $W$ work in $R$ rounds and $S$ span, then
  \begin{enumerate}[leftmargin=12pt]
  \item the initial run of the algorithm with tracking takes $O(W)$ work in expectation and $O(S + R \cdot \log(W))$ span w.h.p.,
  \item running change propagation on a dynamic update to the input configuration $(I', P')$ takes $O(W_\Delta + R')$ work in expectation and $O(S' + R' \log(W'))$
  span w.h.p., where $S', R', W'$ are the maximum span, rounds, and work of the algorithm on the two input configurations,
  \end{enumerate}
  These bounds hold on the \crcwpram{} and in the binary forking \mpram{} model.
\end{theorem}

\begin{proof}
  We begin by analyzing the initial run. By definition, all executions of the round computations,
  \textproc{ComputeRound}, take $O(W)$ work and $O(S)$ span in total, with at most an additional
  $O(\log(n)) = O(\log(W))$ span to perform the parallel for loop. We will show that all additional
  work can be charged to the round computations, and that at most an additional $O(\log(W))$ span
  overhead is incurred.
  
  We observe that $R_{r,p}, W_{r,p}$ and $X_{r,p}$ are at most the size of the work performed by the corresponding
  computations, hence the cost of Lines \ref{code:run-R_r,p} -- \ref{code:run-X_r,p} can be charged to
  the computation. The reader sets $S_m$ can be implemented as dynamic arrays with lazy deletion (this will be discussed
  during change propagation). To append new elements to $S_m$ (Line~\ref{code:run-subscribe}), we can
  use a semisort performing linear work in expectation to first bucket the shared memory locations in $\cup_{p \in P} R_{r,p}$, whose
  work can be charged to the corresponding computations that performed the reads. This adds an additional
  $O(\log(W))$ span w.h.p.\ since the number of reads is no more than $W$ in total.
  Finally, removing retired computations from $P$ (Line~\ref{code:run-filter-P} requires a compaction operation.
  Since compaction takes linear work, it can be charged to the execution of the corresponding
  processes. The span of compaction is at most $O(\log(W))$ in all models.
  
  Summing up the above, we showed that all additional work can be charged to the round computations,
  and the algorithm incurs at most $O(\log(W))$ additional span per round w.h.p. Hence the cost of the
  initial run is $O(W)$ work in expectation and $O(S + R \cdot \log(W))$ span w.h.p.

  We now analyze the change propagation procedure (Algorithm~\ref{alg:round-based-change-prop}).
  The core of the work is the re-execution of the affected readers on Line~\ref{code:change-reexecute},
  which, by definition takes $O(W_\Delta)$ work, and $O(S')$ span, with at most $O(\log(W'))$ additional
  span to perform the parallel for loop. Since some rounds may have no affected computations, the algorithm could perform up to $O(R')$ additional work to process these rounds. We will show that all additional work can be charged to the affected computations, and that no operation incurs more than an additional $O(\log(W'))$ span.

  Lines~\ref{code:change-A} -- \ref{code:change-bucket-affected-computations} bucket the newly
  affected computations by round. This can be achieved with an expected linear work semisort
  and by maintaining the $A_r$ sets as dynamic arrays. The work is chargeable to the affected computations
  and the span is at most $O(\log(W'))$ w.h.p. Computing the current set of affected computations (Line~\ref{code:change-P})
  requires a filter/compaction operation, whose work is charged to the affected computations and span
  is at most $O(\log(W'))$.
  
  Updating the reader sets $S_m$ (Line~\ref{code:change-delete-Sm}) can be done as follows. We maintain $S_m$ as dynamic arrays with lazy
  deletion, meaning that we delete by marking the corresponding slot as empty. When more than half
  of the slots have been marked empty, we perform compaction, whose work is charged to the updates
  and whose span is at most $O(\log(W'))$. In order to perform deletions in constant time, we
  augment the set $R_{r,p}$ so that it remembers, for each entry $m$, the location of $(r,p)$ in $S_m$.
  Therefore these updates take constant amortized work each (using a dynamic array), charged to the
  corresponding affected computations, and at most $O(\log(W'))$ span if a resize/compaction is triggered.
  
  $X^\text{prev}$ can be implemented as an array of size $|P|$, with work charged to the affected
  computations in $P$. As in the initial run, the cost of updating $R_{r,p}, W_{r,p}$ and $X_{r,p}$ can also be
  charged to the work performed by the affected computations.
  
  Updating the reader sets $S_m$ (Line~\ref{code:change-insert-Sm}) is a matter of appending to dynamic
  arrays, and, as mentioned earlier, remembering for each $m \in R_{r,p}$ $m$, the location of $(r,p)$ in $S_m$.
  The work performed can be charged to the affected computations, and the additional span is at most $O(\log(W'))$.

  Collecting the updated locations $U$ (Line~\ref{code:change-U}) can similarly be charged to the affected computations, and incurs no more than $O(\log(W'))$ span.
  On Lines \ref{code:change-L'} -- \ref{code:change-D}, the sets $L'$ and $D'$ are computed by a compaction over $P$, whose work is charged to the affected computations in $P$.
  Updating $L$ and $D$ correspondingly requires a compaction operation, whose work is charged to the affected computations in
  $L$ and $D$ respectively. Each of these compactions costs $O(\log(W'))$ span.

  We can finally conclude that all additional work performed by change propagation can be charged to the affected computations,
  and hence to the computation distance $W_\Delta$, while incurring at most $O(\log(W'))$ additional span per round w.h.p.
  Therefore the total work performed by change propagation is $O(W_\Delta + R')$ in expectation and the span is $O(S' + R' \cdot \log(W'))$ w.h.p.
\end{proof}

We now show that for a special class of \roundbased{} algorithms, the span overhead can be reduced. Our dynamic trees algorithm and our other examples all fall into this special case.

\begin{definition}
  A \emph{restricted \roundbased{}} algorithm is a \roundbased{} algorithm such that each round computation performs only
  a constant number of reads and writes, and each shared memory location is read only by a constant number of computations,
  and only in the round directly after it was written.
\end{definition}

\begin{theorem}\label{thm:restricted-bounds}
  Given a restricted \roundbased{} algorithm $A$ with input configuration $(I, P)$ that does $W$ work in $R$ rounds and $S$ span, then
  \begin{enumerate}[leftmargin=12pt]
  \item the initial run of the algorithm with tracking takes $O(W)$ work and $O(S + R \cdot \overhead{}(W))$ span,
  \item running change propagation on a dynamic update to the input configuration $(I', P')$ takes $O(W_\Delta)$ work and $O(D' + R' \overhead{}(W'))$
  span, where $D', R', W'$ are the maximum span, rounds, and work of the algorithm on the two input configurations,
  \end{enumerate}
  where $\overhead{}(W)$ is the cost of compaction, which is at most
  \begin{enumerate}[leftmargin=12pt]
    \item $O(\log^*(W))$ w.h.p.\ on the \crcwpram{},
    \item $O(\log(W))$ on the binary forking \mpram{}.
  \end{enumerate}
  Furthermore, the work bounds are only randomized (in expectation) on the \crcwpram{}.
\end{theorem}

\begin{proofsketch}
Rather than recreate the entirety of the proof of Theorem~\ref{thm:generic_change_prop_bounds}, we will simply sketch the differences.
In essence, we obtain the result by removing the uses of scans, and semisorts, which were the main cause of the $O(\log(W'))$ span overhead and the randomized work. Instead, we rely only on (possibly approximate) compaction, which is only randomized on the \crcwpram{}, and takes $O(\overhead{}(W'))$ span. We also lose the $R'$ term in the
work since computations can only read from locations written in the previous round,
and hence the set of rounds on which there exists an affected computation must be
contiguous.

The main technique that we will make use of is the sparse array plus compaction technique. In situations where we wish to collect a set
of items from each executed process, we would, in the unrestricted model, require a scan, which costs $O(\log(W'))$ span on the \crcwpram{}.
If each executed process, however, only produces a constant number of these items, we can allocate an array that is a constant size larger
than the number of processes, and each process can write its set of items to a designated offset. We can then perform (possibly approximate)
compaction on this array to obtain the desired set, with at most a constant factor additional blank entries. This takes $O(\overhead{}(W'))$
span.

Maintaining $S_m$ in the initial run and during change propagation is the first bottleneck, originally requiring a semisort. Since each
computation performs a constant number of writes, we can collect the writes using the sparse array plus compaction technique. Since, in
the restricted model, each modifiable will only be read by a constant number of readers, we can update $S_m$ in constant time.

To compute the affected computations $A_r$ also originally required a semisort, but in the restricted model, since all reads happen on the round directly after the write,
no semisort is needed, since they will all have the same value of $r$. Collecting the affected computations from the written modifiables can also be achieved using the
sparse array and compaction technique, using the fact that each computation wrote to a constant number of modifiables, and each modifiable is subsequently read by a constant
number of computations. Additionally, $A_r$ will be empty at the beginning of round $r$, so computing $P$ requires only a compaction operation.

Lastly, collecting the updated locations $U$ can also be performed using the sparse array and compaction technique. In summary, we can replace all
originally $O(\log(W'))$ span operations with $O(\overhead{}(W'))$ equivalents in the restricted setting, and hence we obtain a span bound of
$O(S' + R' \cdot \overhead{}(W'))$ for both the initial run and change propagation.
\end{proofsketch}

\begin{remark}[Space usage]
  We do not formally specify an implementation of the memory model, but one simple way to achieve good space bounds is to use hashtables to implement global shared memory. Each write to a particular global shared memory location maps to an entry in the hashtable. When a round computation is invalidated during a dynamic update, its writes can be purged from the hashtable to free up space, preventing unbounded space blow up. Since the algorithm must also track the reads of each global shared memory location, using this implementation, the space usage is proportional to the number of shared memory reads and writes. In the restricted \roundbased{} model, the number of reads must be proportional to the number of writes, and hence the space usage is optimal, since any strategy for storing shared memory must use at least this much.
\end{remark}

  \newcommand{\tcbuild}{\textbf{$\textproc{Build}$}}
\newcommand{\tcaddvertices}{\textbf{$\textproc{AddVertices}$}}
\newcommand{\tcremovevertices}{\textbf{$\textproc{RemoveVertices}$}}
\newcommand{\tcbatchinsert}{\textbf{$\textproc{BatchLink}$}}
\newcommand{\tcbatchdelete}{\textbf{$\textproc{BatchCut}$}}
\newcommand{\tcbatchisconnected}{\textbf{$\textproc{BatchConnected}$}}
\newcommand{\tcfindrep}{\textbf{$\textproc{FindRepr}$}}
\newcommand{\tcisconnected}{\textbf{$\textproc{IsConnected}$}}

\section{Dynamizing Tree Contraction}\label{sec:tree-contraction}

In this section, we show how to obtain a dynamic tree contraction algorithm by
applying our automatic dynamization technique to the static tree contraction algorithm of Miller and Reif \cite{miller1985parallel}. We will then, in Section~\ref{sec:rc-trees},
describe how to use this to obtain a powerful parallel batch-dynamic trees framework.

\myparagraph{Tree contraction} \emph{Tree contraction} is the process of shrinking a tree down to a single vertex by repeatedly performing local contractions. Each local contraction deletes a vertex and merges its adjacent edges if it had degree two. Tree contraction has a number of useful applications, studied
extensively in~\cite{miller1989parallel,miller1991parallel,acar2005experimental}.
It can be used to perform various computations by associating data
with edges and vertices and defining how data is accumulated during
local contractions.

Various versions of tree contraction have been proposed depending on the specifics of local contractions.
We consider an undirected variant of the randomized version proposed by Miller and Reif \cite{miller1985parallel}, which makes use of two operations: \emph{rake} and \emph{compress}. 
The former removes all nodes of degree one from the tree, except in the case of a pair of adjacent degree one vertices, in which case only one of them is removed by tiebreaking on the vertex IDs. The latter operation, compress, removes an independent set of vertices of degree two that are not adjacent to any vertex of degree one. Compressions are randomized with coin flips to break symmetry. Miller and Reif showed
that it takes $O(\log n)$ rounds w.h.p. to fully contract
a tree of $n$ vertices in this manner.

\myparagraph{Input forests} The algorithms described here operate on undirected forests $F = (V,E)$, where $V$
is a set of vertices, and $E$ is a set of undirected edges. If $(u,v) \in E$, we say that $u$ and $v$ are
\emph{adjacent}, or that they are \emph{neighbors}. A vertex with no neighbors is said
to be \emph{isolated}, and a vertex with one neighbour is called a \emph{leaf}.

We assume that the forests given as input have bounded
degree. That is, there exists some constant $t$ such that each vertex
has at most $t$ neighbors. We will explain how to handle arbitrary-degree trees
momentarily.

\myparagraph{The static algorithm} The static tree contraction algorithm (Algorithm~\ref{alg:forest_contraction-process-inverted}) works in rounds, each of which takes a forest from
the previous round as input and produces a new forest for the next round.
On each round, some vertices may be \emph{deleted}, in which case they
are removed from the forest and are not present in all remaining rounds.
Let $F^i = (V^i, E^i)$ be the forest after $i$ rounds of contraction, and thus $F^0 = F$ is the input forest.
We say that a vertex
$v$ is \emph{alive} at round $i$ if $v \in V^i$, and is \emph{dead} at round $i$
if $v \not\in V^i$. If $v \in V^i$ but $v \not\in V^{i+1}$ then 
$v$ was deleted in round $i$. There are three ways for a vertex to
be deleted: it either \emph{finalizes}, \emph{rakes}, or \emph{compresses}.
Finalization removes isolated vertices. Rake removes all leaves from the tree, with one special exception. If two leaves are adjacent, then to break symmetry and ensure that only one of them rakes, the one with the lower identifier rakes into the other. Finally, compression removes an independent set of degree two vertices that are not adjacent to any degree one vertices, as in Miller and Reif's algorithm.
The choice of which vertices are deleted in each round is made locally for each vertex based upon its own degree, the degrees of its neighbors, and coin flips for itself and its neighbors.
As in the list contraction, for coin flips, we assume a function $\textsc{Heads}(i,v)$ which
indicates whether or not vertex $v$ flipped heaps on round $i$. 
It is important that $\textsc{Heads}(i,v)$ is a function of both the vertex and the round number, as coin flips must be repeatable for change propagation to be correct.

The algorithm produces a \emph{contraction data structure} which serves as a record of the contraction process.
The contraction data structure is a tuple, $(A, D)$, where $A[i][u]$ is a list of pairs containing the vertices adjacent to $u$ in round $i$, and the positions of $u$ in the adjacency lists of the adjacent vertices.  $D[u]$ stores the round on which vertex $u$ contracted. The algorithm also records leaf$[i][u]$, which is true if vertex $u$ is a leaf at round $i$.
%




\myparagraph{Updates}
We consider update operations that implement the interface of a batch-dynamic tree
data structure. This requires supporting batches of links and cuts. A
\emph{link (insertion)} connects two trees in the forest by a newly inserted edge. A
\emph{cut (deletion)} deletes an edge from the forest, separating a
single tree into two trees. We formally specify the interface
for batch-dynamic trees and give a sample implementation of their operations
in terms of the tree contraction data structure in Appendix~\ref{appendix:tree-ops}.

\myparagraph{Handling trees of arbitrary degree} To handle trees of arbitrary degree,
we can split each vertex into a path of vertices, one for each of its neighbors. This technique is
standard and has been described in \cite{johnson1992optimal}, for example. This results in an underlying tree of degree $3$, with at most $O(n+m)$
vertices and $O(m)$ edges for an initial tree of $n$ vertices and $m$ edges. For edge-weighted trees,
the additional edges can be given a suitable identity or null weight to ensure that query values
remain correct.  It is
simple to maintain such a transformation dynamically. When performing
a batch insertion, a work-efficient semisort can be used to group each new neighbour by their endpoints, and
then for each vertex, an appropriate number of new vertices can be added to the path. Batch deletion
can be handled similarly.

  \subsection{Algorithm}

An implementation of the tree contraction algorithm in our framework is shown in Algorithm~\ref{alg:forest_contraction-process-inverted}.

\begin{algorithm}[h!]
  \caption{Tree contraction algorithm}
  \label{alg:forest_contraction-process-inverted}
  \scriptsize
  \begin{algorithmic}[1]		
   \Procedure{ComputeRound}{$i, u$}
     \Read{$A[i][u]$, leaf$[i][u]$}{$((v_1, p_1), ..., (v_t, p_t))$, $\ell$}
       \If{$v_i = \algnull \forall i$}
         \State \textsc{DoFinalize}($i, u$)
       \ElsIf{$\ell$}
         \State \alglocal $(v, p)$ \algassign{} $(v_i, p_i)$ such that $v_i \neq \algnull{}$
         \Read{leaf$[i][v]$}{$\ell'$}
         \If{$\lnot \ell' \lor\ u < v$}
         \textsc{DoRake}($i, u, (v, p)$)
         \Else\ \textsc{DoAlive}($i, u, ((v_1, p_1), ..., (v_t, p_t))$)
         \EndIf
       \Else
         \If{$\exists (v, p), (v', p') : \{ v_1, ..., v_t \} \setminus \{ \algnull \} = \{ v, v' \}$}
         \Read{leaf$[i][v]$, leaf$[i][v']$}{$\ell', \ell''$}
           \State \alglocal $c$ \algassign $\textsc{Heads}(i, u)\ \land\ \lnot\textsc{Heads}(i, v)\ \land\ \lnot\textsc{Heads}(i, v')$
           \If{($\lnot \ell' \land\ \lnot \ell'' \land\ c$)}
             \State \textsc{DoCompress}($i, u, (v, p), (v', p')$)
           \Else\ 
             \State \textsc{DoAlive}($i, u, ((v_1, p_1), ..., (v_t, p_t))$)
           \EndIf
         \EndRead
         \Else
         \State \textsc{DoAlive}($i, u, ((v_1, p_1), ..., (v_t, p_t))$)
         \EndIf
       \EndIf
     \EndRead
   \EndProcedure
    \State
    \Procedure{DoRake}{$i, u, (v, p)$}
      \State \algwritemod{$A[i+1][v][p]$}{\algnull}
      \State \algwritemod{$D[u]$}{$i$}
      \State \algretire
    \EndProcedure
    \State
    \Procedure{DoFinalize}{$i, u$}
      \State \algwritemod{$D[u]$}{$i$}
      \State \algretire
    \EndProcedure
    \State
    \Procedure{DoCompress}{$i, u, (v, p), (v', p')$}
      \State \algwritemod{$A[i+1][v][p]$}{$(v', p')$}
      \State \algwritemod{$A[i+1][v'][p']$}{$(v, p)$}
      \State \algretire
    \EndProcedure
    \State
    \Procedure{DoAlive}{$i, u, ((v_1, p_1), ..., (v_t, p_t))$}
      \State \alglocal nonleaves \algassign $0$
      \For{$j \algassign 1\ \algto t$}
        \If{$v_j \neq \algnull$}
          \State \algwritemod{$A[i+1][v_j][p_j]$}{$(u, j)$}
          \State nonleaves \texttt{+=} 1 - \algorithmicread(leaf$[i][v_j]$)
        \Else\ 
          \State \algwritemod{$A[i+1][u][j]$}{\algnull}
        \EndIf
      \EndFor
      \State \algwritemod{leaf$[i+1][u]$}{nonleaves $= 1$}
    \EndProcedure
  \end{algorithmic}
\end{algorithm}

  \subsection{Analysis}

We analyse the initial work, round, complexity, span, and computation distance of the tree contraction algorithm to obtain bounds for building and updating a parallel batch-dynamic trees data structure. Proofs are given in Appendix~\ref{appendix:tree-contraction-proofs}.

\begin{theorem}\label{thm:tc-costs}
  Given a forest of $n$ vertices, the initial work of tree contraction is $O(n)$ in expectation, the round complexity
  and the span is $O(\log(n))$ w.h.p.\ and the computation distance induced by updating $k$ edges is $O(k\log(1 + n/k))$ in expectation.
\end{theorem}

  \section{Parallel Rake-compress Trees}\label{sec:rc-trees}

Dynamic trees typically provide support for
dynamic connectivity queries.
Most dynamic tree data structures also support some form of
augmented value query. For example, Link-cut trees~\cite{sleatorta83}
support root-to-vertex path queries, and Euler-tour
trees~\cite{henzingerki99} support subtree sum queries.
Top trees~\cite{tarjan2005self,alstrup2005maintaining} support
both path and subtree queries, as well as non-local queries such
as centers and medians, but are not amenable to parallelism.
The only existing parallel batch-dynamic tree data structure is
that of Tseng et al.~\cite{tseng2018batch}, which is based on Euler-tour trees, and
hence only handles subtree queries, and only when the associative operation
is invertible.

Rake-compress trees~\cite{acar2005experimental} (RC trees) are another sequential
dynamic trees data structure, based on tree contraction, and
have also been shown to be capable of handling both path and subtree
queries, as well as non-local queries, all in $O(\log(n))$ time.
In this section, we will explain how our parallel batch-dynamic
algorithm for tree contraction can be used to derive a parallel
batch-dynamic version of RC trees, leading to the first work-efficient algorithm for
batch-dynamic trees that can handle this wide range of queries. We use a slightly different
set of definitions than the original presentation of RC trees in \cite{acar2005experimental},
which correct some subtle corner cases and simplify the exposition,
although the resulting data structure is the same, and hence all of
the query algorithms for sequential RC trees work on our parallel
version.

\myparagraph{Contraction and clusters} RC trees are based
on the idea that the tree contraction process can be interpreted as
a recursive clustering of the original tree. Formally, a
\emph{cluster} is a connected subset of vertices and edges of the original
tree. Note, importantly, that a cluster may contain an edge without containing
both of its endpoints. The \emph{boundary} vertices of a cluster $C$ are the
vertices $v \notin C$ that are adjacent to an edge $e \in C$. The \emph{degree}
of a cluster is the number of boundary vertices of that cluster.
The vertices and edges of the original tree form the base clusters.
Clusters are merged using the following simple rule:
Whenever a vertex $v$ is deleted, all of the clusters that have
$v$ as a boundary vertex are merged with the base cluster containing $v$.
We can therefore see that all clusters formed will have degree at most two.
A cluster of degree zero is called a \emph{nullary} cluster, a cluster of degree one a 
\emph{unary} cluster, and a cluster of degree two a \emph{binary} cluster.
All non-base clusters have a unique \emph{representative vertex},
which corresponds to the vertex that was deleted to form it. For additional clarity, we
provide figures in Appendix~\ref{appendix:rc-trees} that explain what each kind of formed
cluster looks like in more detail.

\subsection{Building and maintaining RC trees}

Given a tree and an execution of the tree contraction algorithm, the RC tree consists
of \emph{nodes} which correspond to the clusters formed by the contraction process.
The children of a node are the nodes corresponding to the clusters that merged together to form it. An example tree, a clustering, and the
corresponding RC tree are depicted in Figure~\ref{fig:rc-tree}. Note that in the case of a disconnected forest, the RC tree will have multiple roots.

\newcommand{\datatype}[1]{\textit{\textbf{#1}}}

We will sketch here how to maintain an RC tree subject to batch-dynamic updates in parallel
using our algorithm for parallel batch-dynamic tree contraction. This requires just two simple
augmentations to the tree contraction algorithm. Recall that tree contraction
(Algorithm~\ref{alg:forest_contraction-process-inverted}) maintains an adjacency list for
each vertex at each round. Whenever a neighbour $u$ of a vertex $v$ rakes into $v$, the
process $u$ writes a null value into the corresponding position in $v$'s adjacency list.
This process can be augmented to also write, in addition to the null value, the identity
of the vertex that just raked. We make one additional augmentation. When storing the data for a neighboring
edge in a vertex's adjacency list, we additionally write the name of the representative vertex if that
edge corresponds to a compression, or null if the edge is an edge of the original tree.
The RC tree can then be inferred using this augmented data as follows.
\begin{enumerate}[leftmargin=12pt]
  \item Given any cluster $C$ with representative $v$, its unary children can be determined
  by looking at the vertices that raked into $v$. The children are precisely the unary clusters
  represented by these vertices. For the final cluster, these are its only children.
  \item Given a binary or unary cluster $C$ with representative $v$, its binary children can be determined
  by inspecting $v$'s adjacency list at the moment it was deleted. The binary clusters corresponding
  to the edges adjacent to $v$ at its time of death are the binary children of the cluster $C$.
\end{enumerate}
It then suffices to observe that this information about the clusters can be recorded during the
contraction process. By employing change propagation, the RC tree can therefore be maintained
subject to batch-dynamic updates. Since each cluster consists of a constant amount of information,
this can be done in the same work and span bounds as the tree contraction algorithm. We therefore
have the following result.
\begin{theorem}
  We can maintain a rake-compress tree of a tree on $n$ vertices subject to batch insertions and batch deletions of
  size $k$ in $O(k\log(1+n/k))$ work in expectation and $O(\log^2(n))$ span per update w.h.p. The span can be improved
  to $O(\log(n) \log^*(n))$ w.h.p. on the \crcwpram{}.
\end{theorem}

\begin{figure}
  \centering
  \begin{subfigure}{0.75\columnwidth}
    \centering
    \includegraphics[width=0.9\columnwidth]{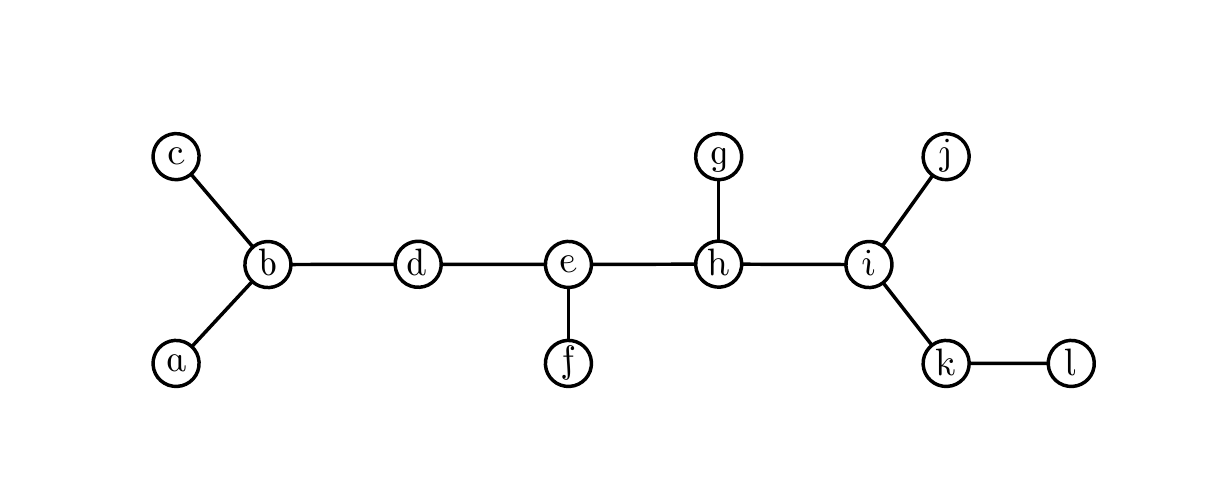}
    \caption{A tree}
  \end{subfigure}
  \begin{subfigure}{0.75\columnwidth}
    \centering
    \includegraphics[width=0.9\columnwidth]{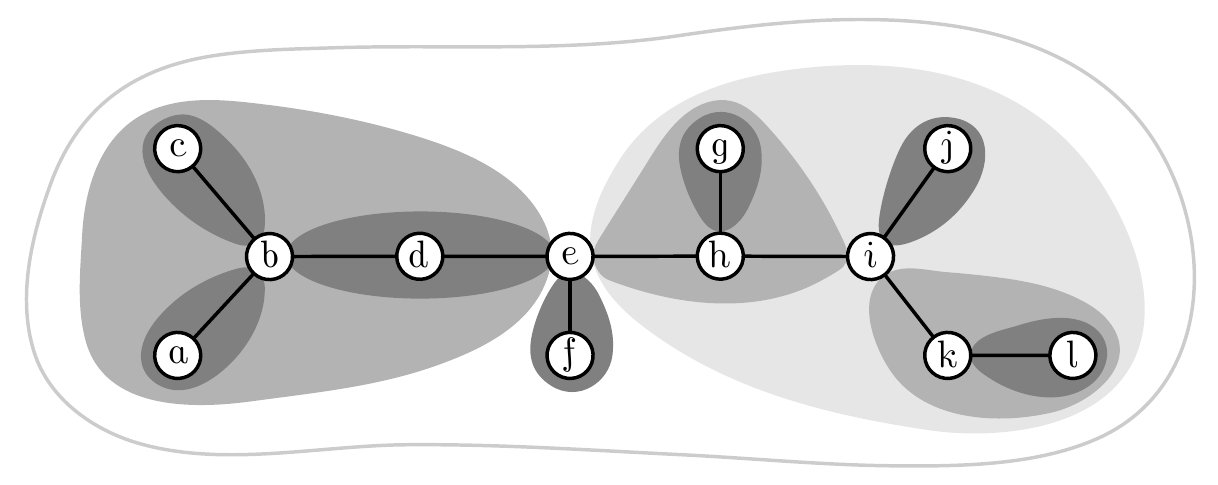}
    \caption{A recursive clustering of the tree produced by tree contraction. Clusters produced in earlier rounds are depicted in a darker color.}
  \end{subfigure}
  \begin{subfigure}{0.9\columnwidth}
    \bigskip
    \centering
    \includegraphics[width=\columnwidth]{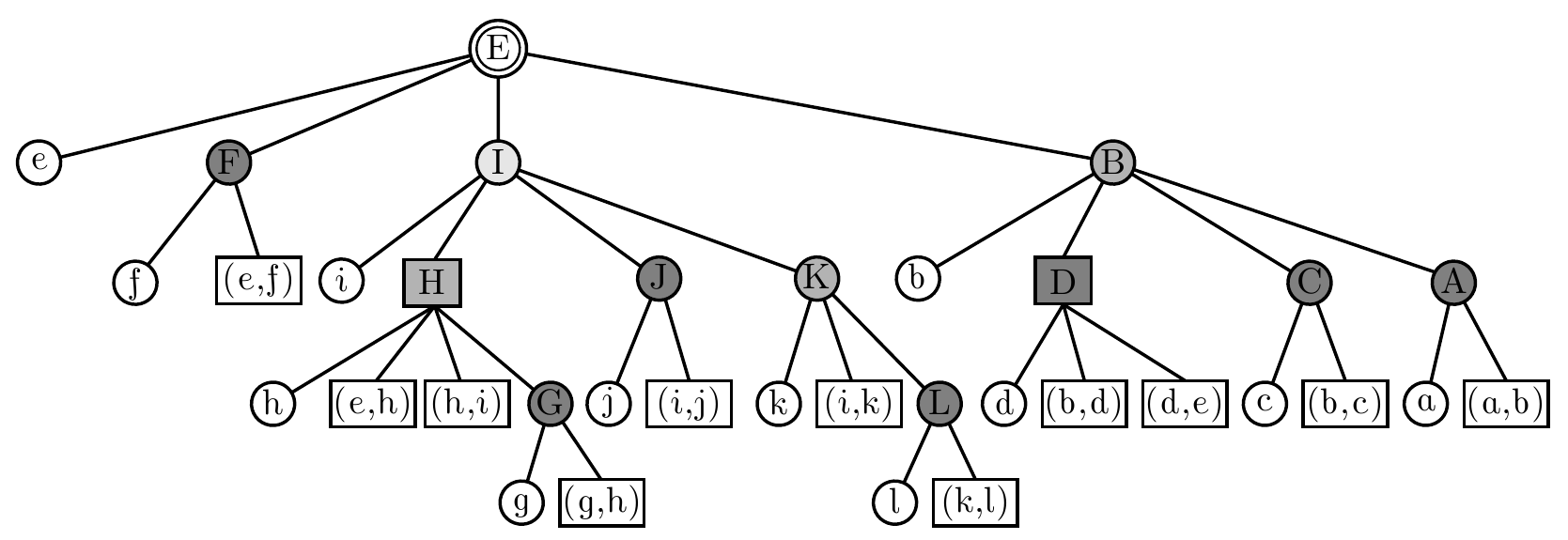}
    \caption{The corresponding RC tree. (Non-base) unary clusters are shown as circles, binary clusters as rectangles, and the finalize (nullary) cluster at the root with two concentric circles. The base clusters (the leaves) are labeled in lowercase, and the composite clusters are labeled with the uppercase of their representative.}
  \end{subfigure}
  \caption{A tree, a clustering, and the corresponding RC tree.}\label{fig:rc-tree}
\end{figure}

\subsection{Applications}

Most kinds of queries assume that the vertices
and/or edges of the input tree are annotated with data, such as weights or labels.
In order to support queries, each cluster is annotated with some additional
information.
The algorithm must then specify how to combine the data from multiple constituent clusters whenever
a set of clusters merge. These annotations are generated
during the tree contraction algorithm, and are therefore available for querying
immediately after performing an update.

Once the clusters are annotated with the necessary data, the queries themselves
typically perform a bottom-up or top-down traversal of the RC tree, or possibly in the 
case of more complicated queries, a combination of both. A variety of application queries is described in \cite{acar2005experimental}.

\myparagraph{Batch queries} For some applications,
we can also implement \emph{batch queries}, in which we answer $k$ queries simultaneously
in $O(k \log(1+n/k))$ work in expectation and $O(\log(n))$ span w.h.p. This improves upon the work bound of $O(k\log(n))$ obtained by simply running the queries in parallel naively. The general idea is to detect when multiple
bottom-up traversals would intersect, and to have only one of them proceed up the RC tree. Upon
reaching the root, the computation can backtrack down the tree in parallel and distribute
the answers to the query. The most obvious query for which this technique is applicable is
finding a representative vertex of the connected component containing a vertex. When traversing
upwards, if multiple query paths intersect, then only one proceeds up the tree and brings the answer
back down for the other ones.
The following theorem formalizes the advantage of performing batch queries. The proof appears in Appendix~\ref{appendix:rc-batch-queries}

\begin{theorem}\label{thm:rc-batch-queries}
  Given a tree on $n$ vertices and a corresponding $RC$ tree, $k$ root-to-leaf
  paths in the RC tree touch $O(k \log(1 + n/k))$ distinct RC tree nodes in expectation.
\end{theorem}

  \section{Conclusion}

In this paper we showed that we can obtain work-efficient parallel batch-dynamic
algorithms by applying an algorithmic dynamization technique to corresponding
static algorithms. Using this technique, we obtained
the first work-efficient parallel algorithm for batch-dynamic trees that supports
more than just subtree queries. It can handle path queries and non-local queries such as centers and medians. Our framework also demonstrates the broad benefits of algorithmic dynamization; much of the complexity of designing parallel
batch-dynamic algorithms by hand is removed, since the static algorithms are usually
simpler than their dynamic counterparts.

We note that although the \roundBased{} model captures a very broad
class of algorithms, the breadth of algorithms suitable for
dynamization is less clear. To be suitable for dynamization, an
algorithm additionally needs to have small computational distance
between small input changes. As some evidence of broad applicability,
however, the practical systems mentioned in the introduction have been applied broadly and successfully---again
without any theoretical justification, yet.

\section*{Acknowledgements}

This work was supported in part by NSF grants CCF-1408940 and CCF-1629444. The authors
would like to thank Ticha Sethapakdi for helping with the RC tree figures.

  \bibliographystyle{abbrv}
  \bibliography{ref,umut}

\begin{thebibliography}{10}

\bibitem{acar2019parallel}
U.~A. Acar, D.~Anderson, G.~E. Blelloch, and L.~Dhulipala.
\newblock Parallel batch-dynamic graph connectivity.
\newblock In {\em {ACM} Symposium on Parallelism in Algorithms and
  Architectures (SPAA)}, 2019.

\bibitem{acar2002adaptive}
U.~A. Acar, G.~E. Blelloch, and R.~Harper.
\newblock Adaptive functional programming.
\newblock In {\em {ACM} Symposium on Principles of Programming Languages
  (POPL)}, 2002.

\bibitem{acar2004dynamizing}
U.~A. Acar, G.~E. Blelloch, R.~Harper, J.~L. Vittes, and S.~L.~M. Woo.
\newblock Dynamizing static algorithms, with applications to dynamic trees and
  history independence.
\newblock In {\em {ACM-SIAM} Symposium on Discrete Algorithms (SODA)}, 2004.

\bibitem{acar2005experimental}
U.~A. Acar, G.~E. Blelloch, and J.~L. Vittes.
\newblock An experimental analysis of change propagation in dynamic trees.
\newblock In {\em Algorithm Engineering and Experiments (ALENEX)}, 2005.

\bibitem{acarcohutu11pd}
U.~A. Acar, A.~Cotter, B.~Hudson, and D.~T{\"u}rko{\u{g}}lu.
\newblock Parallelism in dynamic well-spaced point sets.
\newblock In {\em Symposium on Computational Geometry (SOCG)}, 2011.

\bibitem{alstrup2005maintaining}
S.~Alstrup, J.~Holm, K.~D. Lichtenberg, and M.~Thorup.
\newblock Maintaining information in fully dynamic trees with top trees.
\newblock {\em {ACM} Transactions on Algorithms (TALG)}, 1(2):243--264, 2005.

\bibitem{BhatotiaWiRoAcPa11}
P.~Bhatotia, A.~Wieder, R.~Rodrigues, U.~A. Acar, and R.~Pasquini.
\newblock Incoop: Map{R}educe for incremental computations.
\newblock In {\em ACM Symposium on Cloud Computing (SoCC)}, 2011.

\bibitem{brent74}
R.~P. Brent.
\newblock The parallel evaluation of general arithmetic expressions.
\newblock {\em J. ACM}, 21(2):201--206, 1974.

\bibitem{Condie+10}
T.~Condie, N.~Conway, P.~Alvaro, J.~M. Hellerstein, K.~Elmeleegy, and R.~Sears.
\newblock Mapreduce online.
\newblock In {\em Symposium on Networked Systems Design and Implementation
  (NSDI)}, 2010.

\bibitem{dhulipala2020parallel}
L.~Dhulipala, D.~Durfee, J.~Kulkarni, R.~Peng, S.~Sawlani, and X.~Sun.
\newblock Parallel batch-dynamic graphs: Algorithms and lower bounds.
\newblock In {\em {ACM-SIAM} Symposium on Discrete Algorithms (SODA)}, 2020.

\bibitem{frederickson1985data}
G.~N. Frederickson.
\newblock Data structures for on-line updating of minimum spanning trees, with
  applications.
\newblock {\em {SIAM} J. on Computing}, 14(4):781--798, 1985.

\bibitem{gil1991towards}
J.~Gil, Y.~Matias, and U.~Vishkin.
\newblock Towards a theory of nearly constant time parallel algorithms.
\newblock In {\em {IEEE} Symposium on Foundations of Computer Science (FOCS)},
  1991.

\bibitem{gu2015top}
Y.~Gu, J.~Shun, Y.~Sun, and G.~E. Blelloch.
\newblock A top-down parallel semisort.
\newblock In {\em {ACM} Symposium on Parallelism in Algorithms and
  Architectures (SPAA)}, 2015.

\bibitem{Gunda+10}
P.~K. Gunda, L.~Ravindranath, C.~A. Thekkath, Y.~Yu, and L.~Zhuang.
\newblock Nectar: Automatic management of data and computation in data centers.
\newblock In {\em USENIX Symposium on Operating Systems Design and
  Implementation (OSDI)}, 2010.

\bibitem{henzingerki99}
M.~R. Henzinger and V.~King.
\newblock Randomized fully dynamic graph algorithms with polylogarithmic time
  per operation.
\newblock {\em J. {ACM}}, 46(4):502--516, 1999.

\bibitem{italiano2019dynamic}
G.~F. Italiano, S.~Lattanzi, V.~S. Mirrokni, and N.~Parotsidis.
\newblock Dynamic algorithms for the massively parallel computation model.
\newblock In {\em {ACM} Symposium on Parallelism in Algorithms and
  Architectures (SPAA)}, 2019.

\bibitem{jaja1992introduction}
J.~J{\'a}J{\'a}.
\newblock {\em An introduction to parallel algorithms}, volume~17.
\newblock Addison-Wesley Reading, 1992.

\bibitem{johnson1992optimal}
D.~B. Johnson and P.~Metaxas.
\newblock Optimal algorithms for the vertex updating problem of a minimum
  spanning tree.
\newblock In {\em International Parallel Processing Symposium (IPPS)}, 1992.

\bibitem{karger2000minimum}
D.~R. Karger.
\newblock Minimum cuts in near-linear time.
\newblock {\em J. {ACM}}, 47(1):46--76, 2000.

\bibitem{karp1990parallel}
R.~M. Karp and V.~Ramachandran.
\newblock Parallel algorithms for shared-memory machines, {H}andbook of
  {T}heoretical {C}omputer {S}cience ({J}. van {L}eeuwen, ed.), 1990.

\bibitem{miller1985parallel}
G.~L. Miller and J.~H. Reif.
\newblock Parallel tree contraction and its application.
\newblock In {\em {IEEE} Symposium on Foundations of Computer Science (FOCS)}.
  IEEE, October 1985.

\bibitem{miller1989parallel}
G.~L. Miller and J.~H. Reif.
\newblock Parallel tree contraction part 1: Fundamentals.
\newblock In {\em Randomness and Computation}, pages 47--72. JAI Press,
  Greenwich, Connecticut, 1989.
\newblock Vol. 5.

\bibitem{miller1991parallel}
G.~L. Miller and J.~H. Reif.
\newblock Parallel tree contraction part 2: Further applications.
\newblock {\em {SIAM} J. on Computing}, 20(6):1128--1147, 1991.

\bibitem{Murray+2013}
D.~G. Murray, F.~McSherry, R.~Isaacs, M.~Isard, P.~Barham, and M.~Abadi.
\newblock Naiad: A timely dataflow system.
\newblock In {\em ACM Symposium on Operating Systems Principles (SOSP)}, 2013.

\bibitem{PengDa10}
D.~Peng and F.~Dabek.
\newblock Large-scale incremental processing using distributed transactions and
  notifications.
\newblock In {\em Symposium on Operating Systems Design and Implementation
  (OSDI)}, 2010.

\bibitem{reif1994dynamic}
J.~H. Reif and S.~R. Tate.
\newblock Dynamic parallel tree contraction.
\newblock In {\em {ACM} Symposium on Parallelism in Algorithms and
  Architectures (SPAA)}, 1994.

\bibitem{simsiri2016work}
N.~Simsiri, K.~Tangwongsan, S.~Tirthapura, and K.-L. Wu.
\newblock Work-efficient parallel union-find with applications to incremental
  graph connectivity.
\newblock In {\em European Conference on Parallel Processing (Euro-Par)}, 2016.

\bibitem{sleatorta83}
D.~D. Sleator and R.~E. Tarjan.
\newblock A data structure for dynamic trees.
\newblock {\em J. Computer and System Sciences}, 26(3):362--391, 1983.

\bibitem{tarjan2005self}
R.~E. Tarjan and R.~F. Werneck.
\newblock Self-adjusting top trees.
\newblock In {\em {ACM-SIAM} Symposium on Discrete Algorithms (SODA)}, 2005.

\bibitem{tseng2018batch}
T.~Tseng, L.~Dhulipala, and G.~Blelloch.
\newblock Batch-parallel {Euler} tour trees.
\newblock In {\em Algorithm Engineering and Experiments (ALENEX)}, 2019.

\bibitem{valiant90}
L.~G. Valiant.
\newblock A bridging model for parallel computation.
\newblock {\em Commun. {ACM}}, 33:103--111, 1990.

\end{thebibliography}

  \clearpage
  \appendix

\section{Map-reduce-based Computations}\label{sec:map-reduce}

To illustrate the framework, we describe a simple, yet powerful
technique that we can implement and analyze. This is the so-called
\emph{map-reduce} technique. A map-reduce algorithm takes as input
a sequence $a_0, a_1, ..., a_{n-2}, a_{n-1}$, a unary function $f$, and an associative
binary operator $\oplus$, and computes the value of
\begin{equation}
f(a_0) \oplus f(a_1) \oplus ... \oplus f(a_{n-2}) \oplus f(a_{n-1})
\end{equation}

Although simple, this technique encompasses a wide range of applications,
from computing sums, where $f$ is the identity function
and $\oplus$ is addition, to more complicated examples such as the Rabin-Karp
string hashing algorithm, where $f$ computes the hash value of a character,
and $\oplus$ computes the hash corresponding to the the concatenation of two
hash values.

An implementation of the map-reduce technique in our framework is shown
in Algorithm~\ref{alg:map_reduce}. For simplicity, assume that the input
size $n$ is a power of two, that the input is stored in $A[0...n-1]$, and the
initial set of processes is $0 ... n-1$. The algorithm proceeds in a bottom-up merging fashion, combining
each adjacent pair of elements with the $\oplus$ operator. When the algorithm
terminates, the result will be stored in $V[R][0]$, where $R$ is the index of
the final round.

\begin{algorithm}[h]
  \caption{Map reduce algorithm}
  \label{alg:map_reduce}
  \scriptsize
  \begin{algorithmic}[1]
    \Procedure{ComputeRound}{$r, p$}
    \If{$r = 0$}
      \State \algwritemod{$V[0][p]$}{$f(A[p])$}
    \Else
      \Read{$V[r-1][2p+1]$, $V[r-1][2p+2]$}{$s_1, s_2$}
      \State \algwritemod{$V[r][p]$}{$s_1 \oplus s_2$}
      \EndRead
    \EndIf
    \If{$p \geq n - 2^r$}
    \algretire
    \EndIf
    \EndProcedure
  \end{algorithmic}
\end{algorithm}

\myparagraph{Application to range queries}
Since the intermediate results of the computation are also preserved, it is possible, using standard techniques,
to use this information to perform range queries on any range of the input. That is, the resulting output of the sum algorithm could be used
to compute range sums, and the output of the Rabin-Karp algorithm could be used to compute the hash of any substring of the input string.

\myparagraph{Analysis} We analyze the initial work, round complexity, and span of the algorithm. We also analyze
the computation induced by dynamic updates.

\begin{theorem}
  Given a sequence of length $n$, and a map function $f$ and an associative operation $\oplus$, both taking $O(1)$ time to compute, the initial work of the map reduce algorithm is $O(n)$, the round complexity and span is $O(\log(n))$, and the computation distance of a dynamic update to $k$ elements is $O(k\log(1+n/k))$.
\end{theorem}

\begin{proof}
  First note that since $f$ and $\oplus$ take constant work, each computation
  performs constant work. In round zero, the work
  is therefore $O(n)$. At each round, half of the processes retire, and therefore
  the total work is at most $O\left( n + n/2 + ... \right) = O(n)$
  as desired.
  Since at each round, half of the processes retire, the total number
  of rounds and the span will be $O(\log(n))$.
  
  We now analyze the computation distance of a dynamic update. The affected computations
  can be thought of as a divide-and-conquer tree, a tree in which each computation at round $r > 0$
  has two children, the computations at round $r - 1$ that wrote the values that it read and
  combined. Updating $k$ elements causes $k$ computations at $r = 0$ to become affected, as well as,
  in the worst case, all ancestors of those computations. 
  
  Consider first, all affected computations that occur in rounds $r < \log(1 + n/k)$. Since there
  are $k$ affected computations at $r = 0$ and each can affect at most one computation in the next
  round, there are at most $O(k \log(1+n/k))$ affected computations in these rounds. Now, consider
  the rounds $r \geq \log(1 + n/k)$. Since the number of live computations halves in each round,
  the number of computations (affected or otherwise) at this round is at most
  \begin{equation}
  \begin{split}
    \frac{n}{2^{\log(1+n/k)}} = O\left( \frac{n}{1+n/k} \right) = O\left( \frac{k}{1 + k/n} \right) = O(k).
  \end{split}
  \end{equation}
  Since the number of live computations continues to halve in each round, the total number
  of computations (affected or otherwise) in all rounds $r \geq \log(1+n/k)$ is $O(k)$. Therefore,
  the total number of affected computations across all rounds is at most
  \begin{equation}
  O(k) + O\left(k\log(1 + n/k)\right) = O(k\log(1+n/k)).
  \end{equation}
  Since each affected computation performs constant work, we can conclude that the computation
  distance of a dynamic update to $k$ elements is $O(k\log(1+n/k))$.
\end{proof}

This completes the analysis of the map-reduce technique. Although simple,
the technique is both common and serves as a useful illustrative example of
the framework, and the steps involved in designing and analyzing an algorithm.
That is, we must first
define the input to the algorithm, and the computations that will be performed
at each round. Then, we must analyze the complexity of the algorithm, which
consists of analyzing the initial work, the round complexity and span, then
finally, and most interestingly, the computation distance induced by
dynamic updates to the input. Importantly, the technique of analyzing
computation distance by splitting the rounds at some threshold, often
$O(\log(1+n/k))$), and then bounding the work done before and after the
threshold is very useful, and is used is both analyses of our other two
applications, which are significantly more complex and technically 
challenging.

\newcommand{\lcbuild}{{\bfseries\textsc{Build}}}
\newcommand{\lcbatchsplit}{{\bfseries\textsc{BatchSplit}}}
\newcommand{\lcbatchjoin}{{\bfseries\textsc{BatchJoin}}}
\newcommand{\lcbatchsetvalue}{{\bfseries\textsc{BatchUpdateValue}}}
\newcommand{\lcgetvalue}{{\bfseries\textsc{QueryValue}}}
\newcommand{\lcbatchgetvalue}{{\bfseries\textsc{BatchQueryValue}}}

\section{Dynamizing List Contraction}\label{sec:list-contraction}

List contraction is a fundamental problem in the study of parallel algorithms \cite{karp1990parallel,jaja1992introduction}. In addition to serving as a canonical solution to the list ranking problem (locating an element in a linked list), it is often considered independently as a classic example of a pointer-based algorithm. In this section, we show how the classic parallel list contraction algorithm can be algorithmically dynamized. By dynamizing parallel list contraction, we obtain a canonical dynamic sequence data structure which supports the same set of operations as a classic data structure, the skip list. Our resulting work bounds match the best known hand-crafted parallel batch-dynamic skip lists in the CRCW PRAM model \cite{tseng2018batch}. Lastly, the data structure can be augmented to support queries with respect to a given associative function.

\newcommand{\spliceout}{splice out}
\newcommand{\splicesout}{splices out}
\newcommand{\splicedout}{spliced out}

\myparagraph{List contraction} The list contraction process takes as input a sequence of nodes that form a collection of linked lists, and progressively contracts each list into a single node.   The contraction process operates in rounds, each of which
\splicesout{} an independent set of nodes from the list. When a node is isolated (has no left or right neighbour), it \emph{finalizes}. To select an independent set of nodes to \spliceout{}, we use the random mate technique, in which each node flips an unbiased coin, and is \splicedout{} if it flips heads but its right neighbour flipped tails.

\myparagraph{The static algorithm} The algorithm produces a \emph{contraction data structure}, which records the contraction process and maintains the information necessary to perform queries. This data structure is encoded as a tuple $(L, R, D)$, where $L[i][u]$ and $R[i][u]$ are the left and right neighbours of $u$ at round $i$, and $D[u]$ is the number of rounds that $u$ remained alive (i.e.\ the round number at which it is deleted).

For coin flips, we assume a function $\textsc{Heads}(i,u)$ which
indicates whether or not node $u$ flipped heaps on round $i$. 
It is important that $\textsc{Heads}(i,u)$ is a function of both the node and the round number, as coin flips must be repeatable for change propagation to be correct.

\myparagraph{Updates}\label{subsec:sequences}
We consider update operations that implement the interface of a batch-dynamic
sequence data structure. This includes operations for \emph{joining}
two sequences together and \emph{splitting} a sequence at a given
element. We specify the operations formally supported by batch-dynamic sequences below.

\myparagraph{Augmented dynamic sequences} The list contraction algorithm can be augmented with support for queries with respect to some associative operator $f : D^2 \to D$. This can be achieved by recording,
for each live node $u$, the sum (with respect to $f$) of the values between $u$ and its current right neighbour. This value is updated whenever a node is spliced
out by summing the values recorded on the two adjacent vertices. Queries between nodes $u$ and $v$ can then be performed by walking up the contraction data structure
until $u$ and $v$ meet, summing the values of the adjacent nodes as they go.


\subsection{Interface for Dynamic Sequences}\label{appendix:seq-ops}

Formally, a batch-dynamic sequence supports the following operations.

\begin{itemize}[leftmargin=12pt]
  \item \lcbatchjoin$(\set{(u_1, v_1), \dots, (u_k, v_k)})$
    takes an array of tuples where the $i$-th tuple is a pointer to
    the last element $u_i$ of one sequence and a pointer to the first
    element $v_i$ of a second sequence. For each tuple, the first
    sequence is concatenated with the second sequence. For any
    distinct tuples $(u_i, v_i)$ and $(u_j, v_j)$ in the input, we
    must have $u_i \not= u_j$ and $v_i \not= v_j$.

  \item \lcbatchsplit$(\set{u_1, \ldots, u_k})$ takes
    an array of pointers to sequence elements and, for each element
    $u_i$, breaks the sequence immediately after $u_i$.
\end{itemize}
Optionally, the following can be included to facilitate augmented
value queries with respect to an associative function $f$:
\begin{itemize}[leftmargin=12pt]
  \item \lcbatchsetvalue$(\set{(u_1, a_1), \dots, (u_k, a_k)})$
    takes an array of tuples where the $i$-th tuple contains a pointer
    to an element $u_i$ and a new value $a_i \in D$ for the element.
    The value for $u_i$ is set to $a_i$ in the sequence.

  \item \lcbatchgetvalue$(\set{(u_1, v_1), \dots, (u_k, v_k)})$
    takes an array of pairs of sequence elements. The return value is
    an array where the $i$-th entry holds the value of $f$ applied
    over the subsequence between $u_i$ and $v_i$, inclusive. For all
    pairs, $u_i$ and $v_i$ must be elements in the same sequence, and
    $v_i$ must appear after $u_i$ in the sequence.
\end{itemize}

\noindent An implementation of the high-level interface for updates in terms of the contraction
data structure is illustrated in Algorithm \ref{alg:list_contraction_interface}.

\begin{algorithm}[h]
  \caption{Dynamic sequence operations}
  \label{alg:list_contraction_interface}
  \scriptsize
  \begin{algorithmic}[1]
    \Procedure{Build}{$S$}
    \ParallelFor{$u \algassign 0$ \algto $n-1$}
    \State \algwritemod{$L[0][u]$}{$S[u]$.\textit{prev}}
    \State \algwritemod{$R[0][u]$}{$S[u]$.\textit{next}}
    \State \algwritemod{$A[u]$}{$S[u]$.\textit{value}}
    \EndParallelFor
    \State \textsc{Run}($[n]$)
    \EndProcedure
    \State
    \Procedure{BatchSplit}{$U = \{ u_1, ..., u_k\}$}
    \ParallelFor{\algeach $u \in U$}
    \State \algwritemod{$L[0][R[0][u]]$}{\algnull}
    \State \algwritemod{$R[0][u]$}{\algnull}
    \EndParallelFor
    \State $M = \cup_{u \in U} (L[0][R[0][u]] \cup R[0][u])$
    \State \textsc{Propagate}($M, \emptyset, \emptyset$)
    \EndProcedure
    \State
    \Procedure{BatchJoin}{$U = \{ (u_1, v_1), ..., (u_k, v_k) \}$}
    \ParallelFor{\algeach $(u,v) \in U$}
    \State \algwritemod{$R[0][u]$}{$v$}
    \State \algwritemod{$L[0][v]$}{$u$}
    \EndParallelFor
    \State $M = \cup_{(u,v) \in U} (L[0][v] \cup R[0][v])$
    \State \textsc{Propagate}($M, \emptyset, \emptyset$)
    \EndProcedure
  \end{algorithmic}
\end{algorithm}

\subsection{Algorithm}

An implementation of the list contraction algorithm in our framework is shown in Algorithm~\ref{alg:list_contraction}. 

\begin{algorithm}[h]
  \caption{List contraction algorithm}
  \label{alg:list_contraction}
  \scriptsize
  \begin{algorithmic}[1]		
    \Procedure{ComputeRound}{$i, u$}
      \Read{$L[i][u], R[i][u]$}{$\ell$, $r$}
        \If{$r \neq\ \algnull$}
          \If{\textsc{Heads}($i, u$) $\land\ \lnot$ \textsc{Heads}($i, r$)} \algcomment{Splice out $u$}
            \State \algwritemod{$L[i+1][r]$}{$\ell$}
            \If{$\ell \neq\ \algnull$} \algwritemod{$R[i+1][\ell]$}{$r$}
            \EndIf
            \State \algwritemod{$D[u]$}{$i$}
            \State \algretire
          \Else\ \textproc{StayAlive}($i, u, \ell, r$) \algcomment{$u$ stays alive}
          \EndIf
        \ElsIf{$\ell =\ \algnull$}  \algcomment{Finalize $u$}
          \State \algwritemod{$D[u]$}{$i$}
          \State \algretire
        \Else\ 
          \textproc{StayAlive}($i, u, l, r$)  \algcomment{$u$ stays alive}
        \EndIf
      \EndRead
    \EndProcedure
    
    \State
    
    \Procedure{StayAlive}{$i, u, \ell, r$}
      \If{$r \neq\ \algnull$}
        \algwritemod{$L[i+1][r]$}{$u$}
      \Else\ 
        \algwritemod{$R[i+1][u]$}{\algnull}
      \EndIf
      \If{$\ell \neq\ \algnull$}
        \algwritemod{$R[i+1][\ell]$}{$u$}
      \Else\ 
        \algwritemod{$L[i+1][u]$}{\algnull}
      \EndIf
    \EndProcedure
  \end{algorithmic}
\end{algorithm}

\subsection{Analysis}

We analyse the initial work, round, complexity, span, and computation distance of the list contraction algorithm to obtain bounds for building and updating a parallel batch-dynamic sequence data structure..

\begin{theorem}\label{thm:lc_bounds}
  Given a sequence of length $n$, the initial work of list contraction is $O(n)$ in expectation, the round complexity and span are $O(\log(n))$ w.h.p., and the computation distance of the changes induced by $k$ modifications is $O(k \log(1+n/k))$ in expectation.
\end{theorem}


\noindent  For the analysis, we denote by $\ell^i_S(u)$, the left neighbour of $u$ at round $i$, and similarly by $r^i_S(u)$, the right neighbour of $u$ at round $i$. We denote the absence of a neighbour by $\algnull$. The sequence of nodes that are alive (have not been spliced out) at round $i$ is denoted by $S^i$, e.g.\ $S^0 = S$.

\subsubsection{Analysis of initial construction}

\begin{lemma}\label{lem:list_shrink}
  For any sequence $S$, there exists $\beta \in (0,1)$ such that $\expct{|S^i|} \leq \beta^i |S|$.
\end{lemma}

\begin{proof}
  Consider an node $u$ of $S^i$ at round $i$. If $u$ is isolated, i.e.\ $\ell^i(u) = r^i(u) = \algnull$, then $u$ is spliced out with probability $1$. Otherwise, if $u$ is a tail, i.e.\ $\ell^i(u) \neq \algnull$ and $r^i(u) = \algnull$, then $u$ is spliced out with probability $0$. In any other case, $u$ is spliced out if it flips heads and its right neighbour flips tails, which happens with probability $1/4$. Therefore in a sequence of $n \geq 2$ nodes, the expected number of nodes that splice out is
  \begin{equation}
  \frac{1}{4}(n-1) = \frac{1}{4}n - \frac{1}{4} = \frac{1}{8}n + \frac{1}{8}n - \frac{1}{4} \geq \frac{1}{8}n.
  \end{equation}
  Therefore, since the only node in a sequence of $1$ node is spliced out with probability $1$, we have
  \begin{equation}
  \expct{|S^{i+1}|} \leq |S^i| - \frac{1}{8}|S^i| = \frac{7}{8}|S^i|.
  \end{equation}
  By induction, we can conclude that
  \begin{equation}
  \expct{|S^i|} \leq \beta^i |S|,
  \end{equation}
  with $\beta = 7/8$.
\end{proof}

\begin{lemma}\label{lem:num-rounds-list}
  In a sequence beginning with $n$ nodes, after $O(\log n)$ rounds, there
  are no nodes remaining w.h.p.
\end{lemma}

\begin{proof}
  For any $c > 0$, consider round $r = (c + 1) \cdot \log_{1/\beta}(n)$. By Lemma~\ref{lem:list_shrink} and Markov's inequality, we have
  \begin{equation}
  \probb{\setsize{S^r} \geq 1} \leq \beta^r n = n^{-c}.
  \end{equation}
\end{proof}

\myparagraph{Proof of initial work, rounds, and span in Theorem~\ref{thm:lc_bounds}}

\begin{proof}
  At each round, the construction algorithm performs $O\left(\setsize{S^i}\right)$ work,
  and so the total cost is $O\left(\sum_i \expct{\setsize{S^i}}\right)$ in
  expectation. By Lemma \ref{lem:list_shrink}, this is $O(\setsize{S}) = O(n)$. The
  round complexity and span bounds follow from Lemma~\ref{lem:num-rounds-list}.
\end{proof}

\subsubsection{Analysis of dynamic updates}

\myparagraph{Affected nodes} Recall the notation of an \emph{affected computation}, that is, a computation $(i, u)$ that must be re-executed after a dynamic update either because a value that it read was modified, or because it retired at a different time. We call an node $u$ affected at round $i$ if the computation $(i,u)$ is affected. We make the simplifying assumption that a computation that becomes affected remains affected until it retires. This actually over counts the number of affected computations.

\myparagraph{Bounding the number of affected nodes} For change propagation to be efficient, we must show that the number of affected computations is small at each round. Intuitively, at each round, each affected node may affect its neighbours, which might suggest that the number of affected nodes grows geometrically. However, because an node must have become affected by one of its neighbours, that neighbour is already affected, and hence only at most two additional nodes can become affected per contiguous range of affected nodes, so the growth is only linear in the number of initially affected nodes. Then, since a constant fraction of nodes are spliced out in each round, the number of affected nodes should shrink geometrically, which should dominate the growth of the affected set. We say that an affected node \emph{spreads} to a neighbouring node in round $i$, if that neighbouring node is not affected in round $i$, but is affected in round $i + 1$.

When considering the spread of affected nodes, we must analyse separately the tails of each sequence, since tails are spliced out deterministically (they are spliced out when they are the last remaining node of their sequence), while all other nodes are spliced out randomly. Let $A^i$ denote the set of affected nodes at round $i$. Let $A^i_S$ and $A^i_{S'}$ denote the set of affected non tail nodes at round $i$ that are alive (have not been spliced out) in $S$ and $S'$ respectively.

\begin{lemma}\label{lem:lc-initially-affected}
  Consider a set of $k$ modifications to the input data, i.e.\ $k$ changes to $L$ or $R$ at round $0$. Then $|A^0| \leq k$.
\end{lemma}

\begin{proof}
  The values of $L$ and $R$ are only read by the node that owns them. Hence there are at most $k$ affected nodes at round $0$.
\end{proof}

\begin{lemma}\label{lem:lc-affected-spread}
  Under a set of $k$ modifications to the input data, at most $2k$ new nodes become affected each round.
\end{lemma}

\begin{proof}
  Since computations only read/write their own values and those corresponding to their neighbours, affectation can only spread to neighbouring nodes. Each initially affected node can therefore spread to its neighbours, and its neighbours to their other neighbour and so on. By Lemma~\ref{lem:lc-initially-affected} there are at most $k$ initially affected nodes, hence at most $2k$ new nodes become affected each round.
\end{proof}

\begin{lemma}\label{lem:lc-affected-tails}
  Under a set of $k$ modifications to the input data, the number of affected tail nodes at any point is at most $k$.
\end{lemma}

\begin{proof}
  Since computations only read/write their own values and those corresponding to their neighbours, affectation can only spread to nodes in the same connected sequence, and since each sequence has one tail, by Lemma~\ref{lem:lc-initially-affected}, at most $k$ tails can become affected.
\end{proof}

\begin{lemma}\label{lem:lc-num-affected}
  Under a set of $k$ modifications to the input data,
  \begin{equation}
  \expct{|A^i_S|} \leq 8k,
  \end{equation}
  and similarly for $A^i_{S'}$.
\end{lemma}

\begin{proof}
  By Lemma~\ref{lem:lc-initially-affected}, $|A^0_S| \leq k$. Non-tail nodes are spliced out whenever they flip heads and their right neighbour flips tail, and hence they are spliced out with probability $1/4$. By Lemma~\ref{lem:lc-affected-spread}, at most $2k$ new nodes become affected in each round, and hence we can write
  \begin{equation}
  \expct{|A^i_S|} \leq \frac{3}{4}\expct{|A^{i-1}_S|} + 2k.
  \end{equation}
  Solving this recurrence, we obtain the bound
  \begin{equation}
  \expct{|A^i_S|} \leq k\left(\frac{3}{4}\right)^i + \sum_{j=0}^{i-1} 2k \left(\frac{3}{4}\right)^j \leq 8k.
  \end{equation}
  The same argument shows that $\expct{|A^i_{S'}|} \leq 8k$.
\end{proof}

\begin{lemma}\label{lem:lc-affected-computations}
  Under a set of $k$ modifications to the input data, $\expct{|A^i|} \leq 17k$
\end{lemma}

\begin{proof}
  Follows from Lemma~\ref{lem:lc-affected-tails}, Lemma~\ref{lem:lc-num-affected}, and the fact that $|A^i| \leq |A^i_S| + |A^i_{S'}| + k$.
\end{proof}

\myparagraph{Proof of computation distance in Theorem~\ref{thm:lc_bounds}}

\begin{proof}
  Consider round $r = \log_{1/\beta}(1 + n/k)$, and split the rounds into two groups, those before $r$, and those after $r$. Consider the rounds before $r$. By Lemma~\ref{lem:lc-affected-computations}, there are $O(k)$ affected computations, and since each computation takes $O(1)$ time, the computation distance is
  \begin{equation}
  O(rk) = O\left(k\log\left(1 + \frac{n}{k}\right)\right),
  \end{equation}
  in expectation. For the rounds after $r$, we assume that all computations are affected and apply Lemma~\ref{lem:list_shrink} to deduce that the computation distance is at most
  \begin{equation}
  \sum_{i \geq r} \beta^i |S| = O\left( \frac{n}{1 + n/k} \right) = O(k),
  \end{equation}
  in expectation. Combining these, we find that the total computation distance is $O(k \log(1 + n/k))$ in expectation. 
\end{proof}

  \section{Interface for Dynamic Trees}\label{appendix:tree-ops}

Formally, batch-dynamic trees support the following operations:
\begin{itemize}[leftmargin=12pt]
  \item \tcbatchinsert$(\set{(u_1, v_1), \ldots, (u_k, v_k)})$ takes a batch of edges and adds them to $F$.  The edges must not create a cycle.

  \item \tcbatchdelete$(\set{(u_1, v_1), \ldots, (u_k, v_k)})$ takes a batch of edges and removes them from
    the forest $F$.
\end{itemize}
It is trivial for us to also support adding and deleting vertices from the forest. Optionally, we can also support queries, such as connectivity queries:
\begin{itemize}[leftmargin=12pt]
  \item \tcbatchisconnected$(\set{\set{u_1,v_1}, \ldots, \set{u_k, v_k}})$ takes an array of tuples representing
    queries. The output is an array where the $i$-th entry returns
    whether vertices $u_i$ and $v_i$ are connected by a path in $F$.
\end{itemize}

\noindent An implementation of the high-level interface for updates in terms of the contraction data structure is depicted in Algorithm~\ref{alg:forest_contraction_meta}.

\begin{algorithm}[h]
  \caption{Dynamic tree operations}
  \label{alg:forest_contraction_meta}
  \scriptsize
  \begin{algorithmic}[1]
    \Procedure{Build}{$V,E$}
    \ParallelFor{\algeach vertex $v \in V$}
    \State \algwritemod{$A[0][v]$}{$\{ u : (u, v) \in E \}$}
    \State \algwritemod{leaf$[0][v]$}{($|A[0][v]| = 1$)}
    \EndParallelFor
    \State \textsc{Run}($|V|$)
    \EndProcedure
    \State
    \Procedure{BatchLink}{$E^+ = \{ (u_1,v_1), ... (u_k, v_k) \}$}
    \State \alglocal $U \algassign \cup_{(u,v) \in E^+} \{ u, v \}$
    \ParallelFor{\algeach vertex $u \in U$}
    \State \algwritemod{$A[0][u]$}{$A[0][u] \cup \{ v : (u,v) \in E^+ \}$}
    \State \algwritemod{leaf$[0][u]$}{($|A[0][u]| = 1$)}
    \EndParallelFor
    \State \alglocal $M = \cup_{u \in U} A[0][u] \cup \{ \text{leaf}[0][u]\ |\ \text{leaf}[0][u] \text{changed} \}$
    \State \textsc{Propagate}($M, \emptyset, \emptyset$)
    \EndProcedure
    \State
    \Procedure{BatchCut}{$E^- = \{ (u_1,v_1), ... (u_k, v_k) \}$}
    \State \alglocal $U \algassign \cup_{(u,v) \in E^-} \{ u, v \}$
    \ParallelFor{\algeach vertex $u \in U$}
    \State \algwritemod{$A[0][v]$}{$A[0][v] \setminus \{ u : (u,v) \in E^- \}$}
    \State \algwritemod{leaf$[0][u]$}{($|A[0][u]| = 1$)}
    \EndParallelFor
    \State \alglocal $M = \cup_{u \in U} A[0][u] \cup \{ \text{leaf}[0][u]\ |\ \text{leaf}[0][u] \text{changed} \}$
    \State \textsc{Propagate}($M, \emptyset, \emptyset$)
    \EndProcedure
  \end{algorithmic}
\end{algorithm}

  \section{Analysis of Dynamized Tree Contraction}\label{appendix:tree-contraction-proofs}

Let $F = (V,E)$ be the set of initial vertices and edges of the input tree,
and denote by $F^i = (V^i, E^i)$, the set of remaining (alive) vertices
and edges at round $i$. We use the term ``at round $i$'' to denote the
beginning of round $i$, and ``in round $i$'' to denote an event that
occurs during round $i$.

For some vertex $v$ at round $i$,
we denote the set of its adjacent vertices by $A^i(v)$,
and its degree with $\delta^i(v) = \left|A^i(v)\right|$.
A vertex is \emph{isolated} at round $i$ if $\delta^i(v) = 0$.
When multiple forests are in play, it will be necessary to disambiguate which
is in focus. For this, we will use subscripts: for example, $\delta^i_F(v)$
is the degree of $v$ in the forest $F^i$, and $E^i_F$ is the set of edges in the
forest $F^i$.

\subsection{Analysis of construction}

We first show that the static tree contraction algorithm is efficient.

\begin{lemma}\label{lem:forest-shrink}
  For any forest $(V,E)$, there exists $\beta \in (0,1)$ such that
  $\expct{\setsize{V^i}} \leq \beta^i \setsize{V}$,
  where $V^i$ is the set of vertices remaining after $i$ rounds of contraction.
\end{lemma}

\begin{proof}
  We begin by considering trees, and then extend the argument to forests. Given a tree $(V,E)$, consider the set $V'$ of vertices after one round of contraction. We would like to show there exists $\beta \in (0,1)$ such that $\expct{\setsize{V'}} \leq \beta \setsize{V}$. If $|V| = 1$, then this is trivial since the vertex finalizes (it is deleted with probability $1$). For $|V| \geq 2$, Consider the following sets, which partition the vertex set:
  \begin{enumerate}[]
    \item $H = \{ v : \delta(v) \geq 3 \}$
    \item $L = \{ v : \delta(v) = 1 \}$
    \item $C = \{ v : \delta(v) = 2 \land \forall u \in A(v), u \notin L \}$
    \item $C' = \{ v : \delta(v) = 2 \} \setminus C$
  \end{enumerate}
  
  \noindent Note that at least half of the vertices in $L$ must be deleted, since all leaves are deleted, except those that are adjacent to another leaf, in which case exactly one of the two is deleted. Also, in expectation, $1/8$ of the vertices in $C$ are deleted. Vertices in $H$ and $C'$ necessarily do not get deleted.
  
  Now, observe that $|C'| \leq |L|$, since each vertex in $C'$ is adjacent to a distinct leaf. Finally, we also have $|H| < |L|$, which follows from standard arguments about compact trees. Therefore in expectation,
  \begin{equation}
  \frac{1}{2}|L| + \frac{1}{8} |C| \geq \frac{1}{4} |L| + \frac{1}{8} |H| + \frac{1}{8} |C'| + \frac{1}{8} |C| \geq \frac{1}{8} |V|
  \end{equation}
  vertices are deleted, and hence
  \begin{equation}
  \expct{\setsize{V'}} \leq \frac{7}{8} \expct{\setsize{V}}.
  \end{equation}
  
  Equivalently, for $\beta = \frac{7}{8}$, for every $i$, we have
  $\expct{\setsize{V^{i+1}}} \leq \beta \setsize{V^i}$,
  where $V^i$ is the set of vertices after $i$ rounds of contraction. Therefore
  $\expct{\setsize{V^{i+1}}} \leq \beta \expct{\setsize{V^i}}$.
  Expanding this recurrence, we have $\expct{\setsize{V^i}} \leq \beta^i \setsize{V}$.
  
  To extend the proof to forests, simply partition the forest into its
  constituent trees and apply the same argument to each tree individually.
  Due to linearity of expectation, summing over all trees yields the desired
  bounds.
\end{proof}

\begin{lemma}\label{lem:num-rounds}
  On a forest of $n$ vertices, after $O(\log n)$ rounds of contraction, there
  are no vertices remaining w.h.p.
\end{lemma}

\begin{proof}
  For any $c > 0$, consider round $r = {(c+1)\cdot\log_{1/\beta}(n)}$. By Lemma
  \ref{lem:forest-shrink} and Markov's inequality, we have
  \begin{equation}
  \probb{\setsize{V^r} \geq 1} \leq \beta^r n = n^{-c}.
  \end{equation}
\end{proof}

\myparagraph{Proof of initial work, rounds, and span in Theorem~\ref{thm:tc-costs}}

\begin{proof}
  At each round, the construction algorithm performs $O\left(\setsize{V^i}\right)$ work,
  and so the total work is $O\left(\sum_i \expct{\setsize{V^i}}\right)$ in
  expectation. By Lemma \ref{lem:forest-shrink}, this is $O(\setsize{V}) = O(n)$.
  The round complexity and the span follow from Lemma~\ref{lem:num-rounds}.
\end{proof}

\subsection{Analysis of dynamic updates}

Intuitively, tree contraction is efficiently dynamizable due to the observation
that, when a vertex locally makes a choice about whether or not to delete,
it only needs to know who its neighbors are, and whether or not its
neighbours are leaves. This motivates the definition of the \emph{configuration} of a vertex $v$ at
round $i$, denoted $\kappa^i_F(v)$, defined as
\[ \kappa^i_F(v) = \begin{cases}
(\{ (u, \ell^i_F(u)) : u \in A^i_F(v) \}), & \text{if $v \in V^i_F$} \\
\text{dead}, &\text{if $v \not\in V^i_F$},
\end{cases}
\]
where $\ell^i_F(u)$ indicates whether $\delta^i_F(u) = 1$ (the
\emph{leaf status} of $u$).

Consider some input forest $F=(V,E)$, and let
$F' = (V, (E \setminus E^-) \cup E^+)$ be the newly desired input after a 
batch-cut with edges $E^-$ and/or a batch-link with edges $E^+$.
We say that a vertex $v$ is \emph{affected} at round $i$ if
$\kappa^i_F(v) \neq \kappa^i_{F'}(v)$.

\begin{lemma}\label{lem:affected_computations_are_affected_vertices}
  The execution in the tree contraction algorithm of process $p$ at round $r$ is an affected computation if and only if $p$ is an affected vertex at round $r$.
\end{lemma}

\begin{proof}
  The code for \textsc{ComputeRound} for tree contraction reads only the neighbours, and corresponding leaf statuses, which are precisely the values encoded by the configuration. Hence if vertex $p$ is alive in both forests the computation $p$ is affected if and only if vertex $p$ is affected. If instead $p$ is dead in one forest but not the other, vertex $p$ is affected, and the process $p$ will have retired in one computation but not the other, and hence it will be an affected computation. Otherwise, if vertex $p$ is dead in both forests, then the process $p$ will have retired in both computations, and hence be unaffected.
\end{proof}

This means that we can bound the computation distance by bounding the number of affected
vertices. First, we show that vertices that are not affected at round $i$ have nice properties, as illustrated by Lemmas \ref{lem:unaffected-neighbors} and \ref{lem:unaffected-behavior}.

\begin{lemma}\label{lem:unaffected-neighbors}
  If $v$ is unaffected at round $i$, then either $v$ is dead at round $i$ in both
  $F$ and $F'$, or $v$ is adjacent to the same set of vertices in both.
\end{lemma}

\begin{proof}
  Follows directly from $\kappa^i_F(v) = \kappa^i_{F'}(v)$.
\end{proof}

\begin{lemma}\label{lem:unaffected-behavior}
  If $v$ is unaffected at round $i$, then $v$ is deleted in round $i$ of $F$
  if and only if $v$ is also deleted in round $i$ of $F'$, and in the
  same manner (finalize, rake, or compress).
\end{lemma}

\begin{proof}
  Suppose that $v$ is unaffected at round $i$. Then by definition it has the same neighbours at round $i$ in both $F$ and $F'$. The contraction process depends only on the neighbours of the vertex, and hence proceeds identically in both cases.
\end{proof}

If a vertex $v$ is not affected at round $i$ but is
affected at round $i+1$, then we say that \emph{$v$ becomes affected in round $i$}. A vertex can become affected in many ways, as enumerated in
Lemma~\ref{lem:become-affected}.

\begin{lemma}\label{lem:become-affected}
  If $v$ becomes affected in round $i$, then at least one of the following holds:
  \begin{enumerate}[leftmargin=12pt]
    \item $v$ has an affected neighbor $u$ at round $i$ which was deleted in that
    round in either $F^i$ or $({F'})^i$.
    \item $v$ has an affected neighbour $u$ at round $i+1$ where
    $\ell^{i+1}_F(u) \neq \ell^{i+1}_{F'}(u)$.
  \end{enumerate}
\end{lemma}

\begin{proof}
  First, note that since $v$ becomes affected, we know $v$ does not get deleted,
  and furthermore that $v$ has at least one child at round $i$. (If $v$ were to
  be deleted, then by Lemma \ref{lem:unaffected-behavior} it would do so in both forests, leading
  it to being dead in both forests at the next round and therefore unaffected.
  If $v$ were to have no children, then $v$ would rake, but we just argued that
  $v$ cannot be deleted).
  
  Suppose that the only neighbors of $v$ which are deleted in round $i$ are
  unaffected at round $i$. Then $v$'s set of children in round $i+1$ is the
  same in both forests. If all of these are unaffected at round $i+1$, then
  their leaf statuses are also the same in both forests at round $i+1$, and
  hence $v$ is unaffected, which is a contradiction. Thus case 2 of the lemma
  must hold. In any other scenario, case 1 of the lemma holds.
\end{proof}

\begin{lemma}\label{lem:leaf-status-change}
  If $v$ is not deleted in either forest in round $i$ and
  $\ell^{i+1}_F(v) \neq \ell^{i+1}_{F'}(v)$, then $v$ is affected at round $i$.
\end{lemma}

\begin{proof}
  Suppose $v$ is not affected at round $i$. If none of $v$'s neighbors are deleted
  in this round in either forest, then $\ell^{i+1}_F(v) = \ell^{i+1}_{F'}(v)$, a
  contradiction. Otherwise, if the only neighbors that are deleted do so via a
  compression, since compression preserves the degree of its endpoints, we
  will also have $\ell^{i+1}_F(v) = \ell^{i+1}_{F'}(v)$ and thus a contradiction.
  So, we consider the case of one of $v$'s children raking. However, since $v$
  is unaffected, we know $\ell^i_F(u) = \ell^i_{F'}(u)$ for each child $u$ of $v$.
  Thus if one of them rakes in round $i$ in one forest, it will also do so in
  the other, and we will have $\ell^{i+1}_F(v) = \ell^{i+1}_{F'}(v)$. Therefore
  we conclude that $v$ must be affected at round $i$.
\end{proof}

Lemmas \ref{lem:become-affected} and \ref{lem:leaf-status-change} give us tools
to bound the number of affected vertices for a consecutive round of contraction:
each affected vertex that is deleted affects its neighbors, and each affected vertex whose leaf status
is different in the two forests at the next round affects its parent. This
strategy actually overestimates which vertices are affected, since case 1 of Lemma
\ref{lem:become-affected} does not necessarily imply that $v$ is affected at
the next round. We wish to show that the number of affected vertices at each
round is not large. Intuitively, we will show that the number of affected
vertices grows only arithmetically in each round, while shrinking geometrically,
which implies that their total number can never grow too large.

Let $A^i$ denote the set of affected vertices at round $i$. We begin by bounding the size of $|A^0|$.

\begin{lemma}\label{lem:initially-affected}
  For a batch update of size $k$, we have $|A^0| \leq 3k$.
\end{lemma}

\begin{proof}
  The computation for a given vertex $u$ at most reads its parent, its children, and if it has a single child, its
  leaf status. Therefore, the addition/deletion of a single edge affects at most $3$ vertices at round $0$. Hence $|A^0| \leq 3k$.
\end{proof}

We say that an affected vertex $u$ \emph{spreads to} $v$ in
round $i$, if $v$ was unaffected at round $i$ and $v$ becomes affected
in round $i$ in either of the following ways:
\begin{enumerate}[leftmargin=12pt]
  \item $v$ is neighbor of $u$ at round $i$ and $u$ is deleted in
  round $i$ in either $F$ or $F'$, or
  \item $v$ is neighbor of $u$ at round $i + 1$ and the leaf status of
  $u$ changes in round $i$, i.e., $\ell_F^{i+1}(v) \neq \ell_{F'}^{i+1}(v)$.
\end{enumerate}

Let $s = |A^0|$. For each of $F$ and $F'$, we now inductively construct $s$
disjoint sets for each round $i$, labeled $A^i_1, A^i_2, \ldots A^i_s$. These
sets will form a partition of $A^i$. Begin by arbitrarily partitioning $A^0$ into $s$ singleton sets, and let
$A^0_1,\ldots,A^0_s$ be these singleton sets. (In other words, each affected
vertex in $A^0$ is assigned a unique number $1 \leq j \leq s$, and is then
placed in $A^0_j$.)

Given sets $A^i_1,\ldots,A^i_s$, we construct sets $A^{i+1}_1,\ldots,A^{i+1}_s$
as follows. Consider some $v \in A^{i+1} \setminus A^i$. By Lemmas
\ref{lem:become-affected} and \ref{lem:leaf-status-change}, there must exist at
least one $u \in A^i$ such that
$u$ spreads to $v$. Since there could be many of these, let $S^i(v)$ be the set
of vertices which spread to $v$ in round $i$. Define
\[ j^i(v) = \begin{cases}
j, &\text{if $v \in A^i_j$} \\
\min_{u \in S^i(v)} \left(j \mathbin{\text{where}} u \in A^i_j\right), &\text{otherwise}
\end{cases} \]
(In other words, $j^i(v)$ is $v$'s set identifier if $v$ is affected at round $i$,
or otherwise the minimum set identifier $j$ such that a vertex
from $A^i_j$ spread to $v$ in round $i$). We can then produce the following for
each $1 \leq j \leq k$:
\[ A^{i+1}_j = \{ v \in A^{i+1}~|~j^i(v) = j \} \]
Informally, each affected vertex from round $i$ which stays affected also stays
in the same place, and each newly affected vertex picks a set to
join based on which vertices spread to it.

We say that a vertex $v$ is a \emph{frontier} at round $i$ if $v$ is affected
at round $i$ and at least one of its neighbors in either $F$ or $F'$ is
unaffected at round $i$. It is easy to show that any frontier at any round is
alive in both forests and has the same set of unaffected neighbors in both at
that round (thus, the set of frontier vertices at any round is the same in
both forests). It is also easy to show that if a vertex $v$ spreads to some other
vertex in round $i$, then $v$ is a frontier at round $i$. We show next that
the number of frontier vertices within each $A^i_j$ is bounded.

\begin{lemma}\label{lem:affected-partition-growth}
  For any $i,j$, each of the following statements hold:
  \begin{enumerate}[leftmargin=12pt]
    \item The subforests induced by $A^i_j$ in each of $F^i$ and $(F')^i$ are
    trees.
    \item $A^i_j$ contains at most 2 frontier vertices.
    \item $|A^{i+1}_j \setminus A^i_j| \leq 2$.
  \end{enumerate}
\end{lemma}

\begin{proof}
  Statement 1 follows from rake and compress preserving connectedness,
  and the fact that if $u$ spreads to $v$ then $u$ and $v$ are neighbors in both
  forests either at round $i$ or round $i+1$.
  We prove statement 2 by induction on $i$, and conclude statement 3
  in the process. At round $0$, each $A^0_j$ clearly contains at most 1
  frontier. We now consider some $A^i_j$. Suppose there is a single frontier vertex $v$ in $A^i_j$.
  If $v$ compresses in one
  of the forests, then $v$ will not be a frontier in $A^{i+1}_j$, but it will
  spread to at most two newly affected vertices which may be frontiers at round $i+1$.
  Thus the
  number of frontiers in $A^{i+1}_j$ will be at most 2, and
  $|A^{i+1}_j \setminus A^i_j| \leq 2$.
  
  If $v$ rakes in one of the forests, then we know $v$ must also rake
  in the other forest (if not, then $v$ could not be a frontier, since its
  parent would be affected). It spreads to one newly
  affected vertex (its parent) which may be a frontier at round $i+1$. Thus the
  number of frontiers in $A^{i+1}_j$ will be at most 1, and
  $|A^{i+1}_j \setminus A^i_j| \leq 1$.
  
  Now suppose there are two frontiers $u$ and $v$ in $A^i_j$. Due to statement
  1 of the Lemma, each of these must have at least one affected neighbor at
  round $i$. Thus if either is deleted, it will cease to be a frontier and may
  add at most one newly affected vertex to $A^{i+1}_j$, and this newly affected vertex
  might be a frontier at round $i+1$. The same can be said if either $u$ or $v$
  spreads to a neighbor due to a leaf status change. Thus the number of
  frontiers either remains the same or decreases, and there are at most 2
  newly affected vertices. Hence statements 2 and 3 of the Lemma hold.
\end{proof}

Now define $A^i_{F,j} = A^i_j \cap V^i_F$, that is, the set of vertices from
$A^i_j$ which are alive in $F$ at round $i$. We define $A^i_{F',j}$ similarly
for forest $F'$.

\begin{lemma}\label{lem:affected-partition-size}
  For every $i,j$, we have
  \begin{equation}
  \expct{\setsize{A^i_{F,j}}} \leq \frac 6 {1 - \beta},
  \end{equation}
  and similarly for $A^i_{F',j}$.
\end{lemma}

\begin{proof}
  Let $F^i_{A,j}$ denote the subforest induced by $A^i_{F,j}$ in $F^i$.
  By Lemma \ref{lem:affected-partition-growth}, this subforest is a tree, and
  has at most 2 frontier vertices. By Lemma \ref{lem:forest-shrink}, if we
  applied one round of contraction to $F^i_{A,j}$, the expected number of vertices
  remaining would be at most $\beta \cdot\mathrm{\bf E}[|A^i_{F,j}|]$. However,
  some of the vertices that are deleted in $F^i_{A,j}$ may not be deleted in
  $F^i$. Specifically, any vertex in $A^i_{F,j}$ which is a frontier or is
  the parent of a frontier might not be deleted. There are at most two frontier
  vertices and two associated parents. By Lemma \ref{lem:affected-partition-growth},
  two newly affected vertices might also be added. We also have
  $|A^0_{F,j}| = 1$. Therefore we conclude the following, which
  similarly holds for forest $F'$: \\
  \begin{equation}
  \expct{\setsize{A^{i+1}_{F,j}}}
  \leq \beta \expct{\setsize{A^i_{F,j}}} + 6
  \leq 6 \sum_{r=0}^\infty \beta^r
  = \frac 6 {1 - \beta}.
  \end{equation}
\end{proof}

\begin{lemma}\label{lem:affected-size}
  For a batch update of size $k$, we have for every $i$,
  \begin{equation}
  \expct{\setsize{A^i}} \leq \frac {36}{1-\beta} k.
  \end{equation}
\end{lemma}

\begin{proof}
  Follows from Lemmas
  \ref{lem:initially-affected} and \ref{lem:affected-partition-size}, and
  the fact that
  \begin{equation}
  \setsize{A^i} \leq \sum_{j=1}^s \left( \setsize{A^i_{F,j}} + \setsize{A^i_{F',j}} \right).
  \end{equation}
\end{proof}

\myparagraph{Proof of computation distance in Theorem~\ref{thm:tc-costs}}

\begin{proof}
  Let $F$ be the given forest and $F'$ be the desired forest. Since each process of tree contraction
  does constant work each round, Lemma~\ref{lem:affected_computations_are_affected_vertices} implies
  that the algorithm does $\bigO{\setsize{A^i}}$ work at each round $i$, so $W_\Delta = \sum_i \setsize{A^i}$.
  
  Since at least one vertex
  is either raked or finalized each round, we know that there are at most $n$
  rounds. Consider
  round $r = \log_{1/\beta}(1 + n/k)$, using the $\beta$ given in
  Lemma \ref{lem:forest-shrink}. We now split
  the rounds into two groups: those that come before $r$ and those
  that come after.
  
  For $i < r$, we bound $\expct{\setsize{A^i}}$ according to
  Lemma \ref{lem:affected-size}, yielding
  \begin{equation}
  \sum_{i < r} \expct{ \setsize{A^i} } = O(rk) = O\left(k\log\left(1 + \frac{n}{k}\right)\right)
  \end{equation}
  work. Now consider $r \leq i < n$. For any $i$ we know
  $\setsize{A^i} \leq \setsize{V^i_F} + \setsize{V^i_{F'}}$, because each affected
  vertex must be alive in at least one of the two forests at that round. We
  can then apply the bound given in Lemma \ref{lem:forest-shrink}, and so
  \begin{equation}
  \begin{split}
  \sum_{r \leq i < n}  \expct{\setsize{A^i}} &\leq  \sum_{r \leq i < n} \left( \expct{\setsize{V^i_F}} + \expct{\setsize{V^i_{F'}}} \right) \\
  &\leq \sum_{r \leq i < n} \left( \beta^i n + \beta^i n \right) \\
  &= O(n \beta^r) \\
  &= O\left( \frac{nk}{n + k} \right) \\
  &= O\left( \frac{k}{1 + \frac{k}{n}} \right) \\
  &= O(k),
  \end{split}
  \end{equation}
  and thus
  \begin{equation}
  \expct{W_\Delta} = O\left(k\log\left(1 + \frac{n}{k}\right) \right) + O(k) = O\left(k\log\left(1 + \frac{n}{k}\right) \right).
  \end{equation}
\end{proof}

  \section{Additional Information on Rake-compress Trees}\label{appendix:rc-trees}

\subsection{Visualizing cluster formation}\label{appendix:rc-clusters}

In an RC tree, clusters are formed whenever a vertex is deleted. Specifically,
when vertex $v$ is deleted, all clusters that have $v$ as a boundary vertex
are merged with the base cluster containing $v$.

\begin{enumerate}[leftmargin=12pt]
\item Whenever a vertex $v$ rakes into a vertex $u$, a unary cluster is formed
that contains the vertex $v$ (a base cluster), the cluster corresponding to
the edge $(u,v)$ (formally, the binary cluster with boundaries $u$ and $v$),
and the clusters corresponding to all of the rakes of vertices $c_1, c_2, ...$ that raked into $v$ (formally, all unary clusters whose boundary is the vertex $v$). The cluster's representative is the vertex $v$, and its boundary is the vertex $u$.

\begin{figure}[H]
  \centering
  \includegraphics[width=0.45\columnwidth]{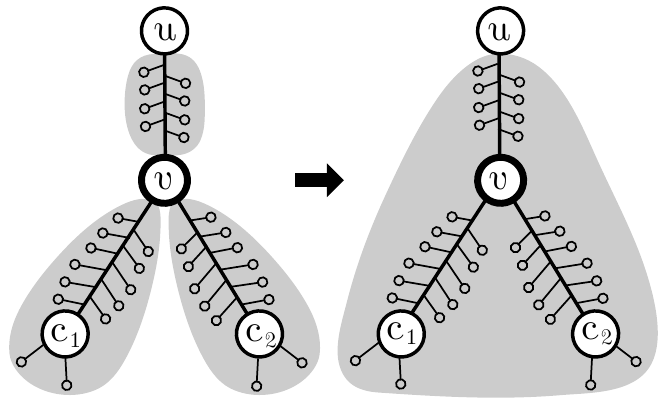}
\end{figure}

\item When a vertex $v$ is compressed between the vertices $u$ and $w$, a binary
cluster is formed that contains the vertex $v$ (a base cluster), the clusters
corresponding to the edges $(u,v)$ and $(v,w)$ (formally, the binary clusters with boundaries $u$ and $v$ and $v$ and $w$ respectively), and the clusters corresponding to all of the rakes of vertices $c_1, c_2, ...$ that raked into $v$ (formally, all unary clusters whose boundary is the vertex $v$). The cluster's representative is the vertex $v$, and its boundaries are the vertices $u$ and $w$.

\begin{figure}[H]
  \centering
  \includegraphics[width=0.65\columnwidth]{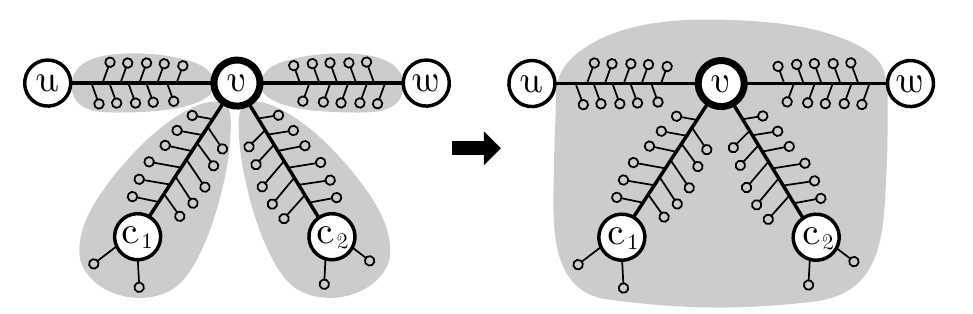}
\end{figure}

\item When a vertex $v$ finalizes, a nullary cluster is formed that contains the
vertex $v$ (a base cluster), and the clusters corresponding to all of the rakes of vertices $c_1, c_2, ...$ that raked into $v$ (formally, all unary clusters whose boundary is the vertex $v$). The cluster's representative is the vertex $v$.

\begin{figure}[H]
  \centering
  \includegraphics[width=0.55\columnwidth]{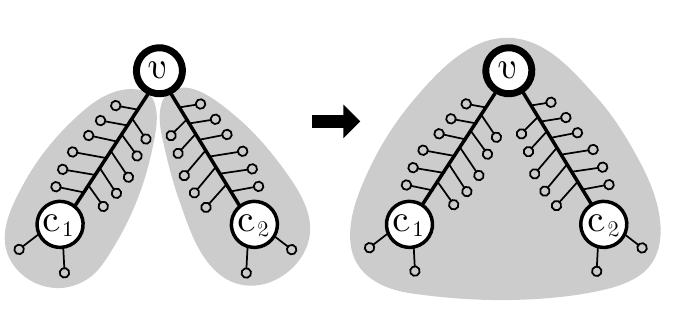}
\end{figure}

\end{enumerate}

\subsection{Proof of Theorem~\ref{thm:rc-batch-queries}}\label{appendix:rc-batch-queries}

\begin{proof}
  The proof is almost identical to that of Theorem~\ref{thm:tc-costs}. First,
  note that we can associate RC tree nodes to rounds by assigning them to
  the round in which their representative was deleted. Consider
  round $r = \log_{1/\beta}(1 + n/k)$, using the $\beta$ given in
  Lemma \ref{lem:forest-shrink}. We bound the number of RC tree nodes touched before
  and after round $r$. Since each root-to-leaf path touches at most one node
  per round, there are at most $O(k \log(1 + n/k))$ touched nodes before round $r$.
  After round $r$, we apply the bound given in Lemma~\ref{lem:forest-shrink} to
  show that the total number of vertices remaining and hence the total number
  of RC tree nodes is at most
  \begin{equation}
  \begin{split}
  \expct{\sum_{i \geq r} |V^i|} &\leq \sum_{i \geq r} \beta_i n  \\
   &= O\left( \frac{nk}{n + k} \right) \\
  &= O\left( \frac{k}{1 + \frac{k}{n}} \right) \\
  &= O(k).
  \end{split}
  \end{equation}
  Therefore, the total number of touched nodes is at most
  \begin{equation}
  O\left(k\log\left(1 + \frac{n}{k}\right) \right) + O(k) = O\left(k\log\left(1 + \frac{n}{k}\right) \right),
  \end{equation}
  as desired.
\end{proof}

\end{document}